%% file: neuralComp.tex
\documentclass[11pt]{article}

\newcommand{\alphaT}[0]{\gamma_T}
\newcommand{\mpl}[0]{Marchenko-Pastur law}

\usepackage{microtype}

\usepackage{amsmath}

\usepackage{times}
\usepackage[breaklinks=true]{hyperref}

\usepackage{graphicx}
\usepackage{color}
\usepackage{xcolor} 
\hypersetup{
	colorlinks,
	linkcolor={red!50!black},
	citecolor={blue!50!black},
	urlcolor={blue!80!black}
}
\usepackage{multirow}
\usepackage[authoryear]{natbib}
\usepackage{breakcites}
\usepackage{rotating}
\usepackage{bbm}
\usepackage{latexsym}

\renewcommand{\thesubsubsection}{\arabic{section}.\arabic{subsubsection}}

\newcommand{\captionfonts}{\normalsize}

\makeatletter  
\long\def\@makecaption#1#2{%
  \vskip\abovecaptionskip
  \sbox\@tempboxa{{\captionfonts #1: #2}}%
  \ifdim \wd\@tempboxa >\hsize
    {\captionfonts #1: #2\par}
  \else
    \hbox to\hsize{\hfil\box\@tempboxa\hfil}%
  \fi
  \vskip\belowcaptionskip}
\makeatother   

\input{commonConfigStuff.tex}

\makeatletter
\renewcommand{\maketitle}{\bgroup\setlength{\parindent}{0pt}
	\begin{flushleft}
	{\huge	\@title}
		
\vspace{.5cm}		
		\@author
	\end{flushleft}\egroup
}
\makeatother

\date{}
\title{
From univariate to multivariate coupling between continuous signals and point processes: \newline
	a mathematical framework
}

\author{ {\large Shervin Safavi$^{1,3}$, Nikos K. Logothetis$^{ 1,4}$, Michel Besserve$^{ 1,  2,*}$.}\\
	\vspace*{.5cm}
	\itshape
	{$^{ 1}$MPI for Biological Cybernetics, T\"ubingen, Germany.}\\
	{$^{ 2}$MPI for Intelligent Systems, T\"ubingen, Germany.}\\
	{$^{ 3}$IMPRS for Cognive and Systems Neuroscience, University of T\"ubingen, Germany.}\\
	{$^{ 4}$University of Manchester, United Kingdom.}\\
	{$^{*}$ Correspondence: \rm \texttt{michel.besserve@tuebingen.mpg.de}.}
}
\begin{document}


\maketitle

%


\thispagestyle{empty}
\markboth{}{NC instructions}

\begin{center} {\bf Abstract} \end{center}
Time series datasets often contain heterogeneous signals, 
composed of both continuously changing quantities and discretely occurring events. 
The coupling between these measurements may provide insights into key underlying mechanisms of the systems under study. 
To better extract this information, we investigate the asymptotic statistical properties of coupling measures between continuous signals and point processes. 
We first introduce martingale stochastic integration theory as a mathematical model for a family of statistical quantities that include the Phase Locking Value, a classical coupling measure to characterize complex dynamics. Based on the martingale Central Limit Theorem, we can then derive the asymptotic Gaussian distribution of estimates of such coupling measure, that can be exploited for statistical testing. Second, based on multivariate extensions of this result and Random Matrix Theory,
we establish a principled way to analyze the low rank coupling between a large number of point processes and continuous signals. For a null hypothesis of no coupling, we establish sufficient conditions for the empirical distribution of squared singular values of the matrix to converge, as the number of measured signals increases, to the well-known Marchenko-Pastur (MP) law, and the largest squared singular value converges to the upper end of the MPs support. This justifies a simple thresholding approach to assess the significance of multivariate coupling.
Finally, we illustrate with simulations the relevance of our univariate and multivariate results in the context of neural time series, addressing how to reliably quantify the interplay between multi channel Local Field Potential signals and the spiking activity of a large population of neurons.

\section{Introduction}
\input{neuralComp_intro.tex}
\section{Background}
\subsection{Spike field coupling in Neuroscience}
\input{neuralComp_spkFieldCouplingBackground.tex}\label{sec:spc}
\subsection{Counting process martingales}\label{sec:cpm}
\input{neuralComp_background_cpm.tex}
\subsection{Random matrix theory}\label{sec:RMTbackground}
\input{neuralComp_background_rmt.tex}

\section{Assessment of univariate coupling}

\subsection{Mathematical formulation}
\input{uniVar_math.tex}

\subsection{Application to bias assessment}\label{ssec:coro3theoPredict}
\input{uniVar_application_v1.tex}

\subsection{Simulations}\label{ssec:simulation_1d}
\input{neuralComp_experimetns_plvDist.tex}

\section{Assessment of multivariate coupling}\label{sec:multiCoupl}
As a natural extension of the scalar case discussed in the previous section,
we now consider the expected coupling matrix 
$\boldsymbol{C}^*$ 
between a $n$-dimensional vector of counting processes $\boldsymbol{N}$ 
with associated intensity vector $\boldsymbol{\lambda}(t)$ 
and a multivariate $p$-dimensional signal $\boldsymbol{x}(t)$, 
and its estimate based on independent trials $\widehat{\boldsymbol{C}}_K$, 
respectively defined as
\begin{equation}
  \label{eq:defEmpMultVarC}
  \boldsymbol{C}^* = \int_0^T \boldsymbol{x}(t)\boldsymbol{\lambda}(t)^\top dt\quad \text{and} \quad \widehat{\boldsymbol{C}}_K = \frac{1}{K}\sum_{k=1}^K\int_0^T \boldsymbol{x}(t)d\boldsymbol{N}^{(k)}(t)^\top\,.
\end{equation}
In this multivariate setting, 
the coupling matrix between the point process and continuous signal can be characterized by the singular value(s) of $\boldsymbol{C}^*$
\[
\sigma_1 \geq \sigma_2 \geq \dots \geq \sigma_{p} \geq 0\,,
\]
and associated orthonormal singular vectors $\{(\boldsymbol{u}_k,\boldsymbol{v}_k)\}$, such that 
\[
\boldsymbol{C}^* = \sum_{k = 1}^p\boldsymbol{u}_k\sigma_k\boldsymbol{v}_k^H\,.
\]
When the dimension of the coupling matrix gets large, recovering the entire structure of $\boldsymbol{C}^*$ using its estimate $\widehat{\boldsymbol{C}}_K$ becomes unlikely due to the fluctuations of individual coupling coefficients investigated in the previous section.  However, the largest singular values may remain reliably estimated because they correspond to a low rank structure of the matrix that stand out from the noise. Random matrix theory provides justifications for this approach by characterizing the spectral properties of ``noisy'' matrices. Up to a normalization explained later, this will involve indirectly characterizing the behavior of the empirical singular vectors $\{\widehat{\sigma}_k\}$ of the estimate matrix $\widehat{\boldsymbol{C}}_K$, by analyzing the the eigenvalues of the hermitian matrix
$\frac{1}{n}{\widehat{\boldsymbol{C}}_K}{\widehat{\boldsymbol{C}}_K}^H$ denoted
\[
\ell_1 \geq \ell_2 \geq \dots \geq \ell_{p} \geq  0\,.
\]
These are related to each other by the relation $\widehat{\sigma}_k = \sqrt{n\ell_k}$ for all $k$.

\subsection{Mathematical formulation}
\input{multVar_math.tex}

\subsection{Application to significance assessment}\label{sec:signifEV}                        
\input{neuralComp_appInSigcAssess.tex}

\subsection{Simulation}\label{ssec:simulation_multD}
\input{neuralComp_experimetns_multVar_v2.tex}


\section{Discussion}\label{sec:discussion}
\input{neuralComp_discussion.tex}

\section*{Conclusion}
\input{neuralComp_conclusion.tex}

\subsection*{Acknowledgments}
We are very grateful to Afonso Bandeira and Asad Lodhia
for fruitful discussions at the beginning of the project.
We thank Edgar Dobriban for pointing us to \citet{baiCentralLimitTheorems2008};
Joachim Werner and Michael Schnabel for their excellent
IT support.
This work was supported by the Max Planck Society.

\bibliographystyle{APA}
\bibliography{ploscb2018_gpla,zotlib}
\newpage
\appendix
\makeatletter
\renewcommand{\thesubsection}{\thesection.\arabic{subsection}}
\renewcommand{\thesubsubsection}{\thesubsection.\arabic{subsubsection}}
\makeatother
\input{neuralComp_appendix.tex}

\end{document}

%% file: commonConfigStuff.tex
\usepackage{nameref}

\usepackage{xspace}
\newcommand{\ie}{i.\,e.\xspace~}

\newcommand{\eg}{e.\,g.\xspace~}

\usepackage{soul}

\usepackage{amsmath,amssymb,amsthm}

\newtheorem{thm}{Theorem}
\newtheorem{corol}{Corollary}
\newtheorem{assum}{Assumption}

\newtheorem{defn}{Definition}
\newtheorem{prop}{Proposition}
\newtheorem{remk}{Remark}

%% file: neuralComp_intro.tex
The observation of highly multivariate temporal point processes,
corresponding to the activity of a large number of individuals or units,
is pervasive in many applications 
(\eg neurons in brain networks \citep{johnson1996point}, 
members in social networks \citep{dai2016recurrent,de2016learning}). 
As the number of observed events per unit may remains small, 
inferring the underlying dynamical properties of the studied system from such observations is challenging. 
However, in many cases, it is possible to observe continuous signals whose coupling with the events can offer key insights. 

In Neuroscience, this is the case of the extracellular electrical field, 
which provides information complementary to spiking activity.
Local Field Potentials (LFP), are mesoscopic \citep{Liljenstroem2012} signals resulting from the superposition of the electric potentials generated by ionic currents flowing across the membranes of the cells located close the tip of recording electrodes.
The LFP reflects neural cooperation due to the anisotropic cytoarchitecture of most brain regions,
allowing the summation of the extracellular currents resulting from the activity of neighboring cells. As such, a number of subthreshold integrative processes 
(\ie modifying the neurons' internal state without necessarily triggering spikes) 
contribute to the LFP signal
\citep{Buzsaki2012,Buzsaki2013sbs,Einevoll2014mal,pesaranInvestigatingLargescaleBrain2018,herrerasLocalFieldPotentials2016}.

Reliably quantifying the \emph{coupling} between activities of individual units 
(\eg spikes generated by individual neurons) in a circuit 
and the aggregated measures (such as the LFP) may provide insights into underlying network mechanisms,
as illustrated in the electrophysiology literature.
At the single neuron level, the relationship of spiking activity to subthreshold activity has broad implications for the underlying cellular and network mechanisms at play.
For instance, it has been suggested that synaptic plasticity triggers changes in the coupling between spikes and LFPs 
\citep{grosmark2012REMSleepReorganizes,grosmark2016DiversityNeuralFiring}.
Regarding the putative functional role of such observed couplings, it has been hypothesized to 
 support cognitive functions such as attention. Such \emph{coordination by oscillations} hypothesis
 proposes that network oscillations modulate differentially the excitability of several target populations, such that a sender population can emit messages during the window of time for which a selected target is active, while unselected targets are silenced \citep{friesRhythmsCognitionCommunication2015,womelsdorfModulationNeuronalInteractions2007a,fries2005MechanismCognitiveDynamics}.


In the case of two continuous signals, 
coupling measures such as coherence and Phase Locking Value ($\rm PLV$) \citep{rosenblum2001phase,pereda2005nonlinear} are widely used and their statistical properties have been investigated, 
in particular in the stationary Gaussian case \citep{brillinger1981time,aydore2013note}. 
In a similar way, $\rm PLV$
\citep{ashidaProcessingPhaseLockedSpikes2010} and Spike-Field Coherence (SFC) \citep{mitra2007observed} 
can measure spike-LFP coupling 
(see among others: \cite{vinck2012ImprovedMeasuresPhasecoupling,vinck2010PairwisePhaseConsistency,jiangMeasuringDirectionalityNeuronal2015,zarei2018IntroducingComprehensiveFramework,liUnbiasedRobustQuantification2016}), 
and are broadly used to makes sense of the role played by neurons in coordinated network activity \citep{buzsaki2015WhatDoesGamma}. There are notable contributions investigating potential biases of those measures, when both point processes and continuous signals are involved \citep{lepage2011dependence,kovach2017BiasedLookPhase}. However, two questions relevant for practical applications remain: (1) the effect of intrinsic variability of spike occurrence on key statistical properties of the estimates, such as the variance, have not yet been thoroughly described; (2) how to extend rigorous statistical analysis of spike-filed coupling  in the context of the highly multivariate signals available with modern recording techniques, remains largely unaddressed. 


We address these two questions by using continuous time martingale theory (see \eg \citet{liptser2013statistics}), the related concept of stochastic integration (see \eg \citep{protter2005stochastic}) 
and Random Matrix Theory \citep{anderson2010introduction,capitaine2016SpectrumDeformedRandom}.
The Martingale Central Limit Theorem (CLT) allows us to derive analytically the asymptotic Gaussian distribution of a general family of coupling measure that can be expressed as stochastic integrals.
We exploit this general result to show that the classical univariate $\rm PLV$ estimator is also asymptotically normally distributed, and provide the analytical expression for its mean and variance. 
Furthermore, we study potential sources of bias for the commonly used \textit{von Mises} coupling model \citep{ashidaProcessingPhaseLockedSpikes2010}.
We then go beyond univariate coupling measures and analyze the statistical properties of a family of multivariate coupling  measures taking the form of a \emph{matrix} with stochastic integral coefficient. 
We characterize the jointly Gaussian asymptotic distribution of matrix coefficients, and exploiting Random Matrix Theory (RMT) principles to show that, after appropriate normalization,
the spectral distribution of such large matrices under the null hypothesis (of absence of coupling),
follows approximately the Marchenko-Pastur (MP) law \citep{marchenko1967distribution}, while the magnitude of the largest singular value converges to fixed value whose simple analytic expression depends only of the shape of the matrix.
We finally show how this result provides a fast and principled procedure to detect significant singular values of the coupling matrix, reflecting an actual dependency between the underlying signals. In the appendices, we included detailed proofs and background material on RMT and stochastic integration, such that non-expert readers can further  apply these tools in Neuroscience.


%% file: neuralComp_spkFieldCouplingBackground.tex
Although our results are relevant to a broad range of applications, 
within and beyond Neuroscience, we will use the estimation of spike-LFP coupling introduced above, as the guiding example of this paper.
Spikes convey information communicated between individual neurons. 
Such communications is believed to be encoded in the occurrence times of successive spike events, which are typically modeled with point processes
(\eg Poisson \citep{softky1993highly} or Hawkes process \citep{truccolo2016PointProcessObservations,kruminCorrelationbasedAnalysisGeneration2010}).

While oscillatory dynamics is ubiquitous in the brain and instrumental to its coordinated activity  \citep{Buzsaki2006,Buzsaki2013sbs,peterson2018HealthyOscillatoryCoordination},
it is often challenging to uncover based solely on the sparse spiking activity of recorded neurons.
On the other hand, LFPs often exhibit oscillatory components that can be isolated with signal processing tools (typically band-pass filtering or template matching), such that pairing the temporal information from LFPs and spiking activity can help extract reliable markers of neural coordination.

An example of coupling measure achieving such pairing 
is the Phase Locking Value ($\rm PLV$).
Given event (spike) times $\{t_j\}$ where $j \in \{1,2,\dots, N\}$ 
(where $N$ is number of spikes in the spike train)
and $\phi(t)$ the time-varying phase of a oscillatory continuous signal
which is typically a band-passed filtered LFP,
phase locking between these signals is estimated by the complex number 
\begin{equation}
  \label{eq:plvDef}
 \widehat{\rm PLV} = \frac{1}{N}\sum_{j = 1}^{N} e^{\boldsymbol{i}\phi(t_j)}\,,\mbox{ with } \boldsymbol{i}^2=-1\,.
\end{equation}
We use a hat notation to reflect that this quantity is \textit{empirical}: 
indeed, even if we assume a fixed $\phi$, 
the $\rm PLV$  depends on the specific values of event times $t_j$. 
In the present work, we will assume these points are drawn from a Poisson process, with a possibly time varying rate (inhomogeneous Poisson process), 
such that we can define a \textit{population} statistics that is a function of the point process population distribution instead of its empirical counterpart. 
We will then address under which conditions the empirical $\rm PLV$ reflects a true coupling between the rate of underlying point process and $\phi$.


%% file: neuralComp_background_cpm.tex
In this paper, we use a continuous time framework leading to powerful results based on concise deterministic and stochastic integral expressions, which can trivially be approximated using discrete time signals in practice. A (continuous time) stochastic process $M = \{M(t); t \in [0, \tau]\}$ is a zero-mean martingale relative\footnote{Any martingale in this paper is zero-mean}
to the filtration $\{\mathcal{F}_t\}$ 
(which represents the past information accumulated up to time~$t$)
if (1)~$M(0)=0$, (2)~it is adapted to $\{\mathcal{F}_t\}$ 
(informally the law of $M$ up to time $t$ ``uses'' only past information up to $t$),
and (3)~it satisfies the martingale property
\begin{equation}\label{eq:martingaleDef}
E\left[M(t)|\mathcal{F}_s\right] = M(s),
\quad
\text{for all}
\quad
t > s\,.
\end{equation}

Consider now a (univariate) counting process
$\left\{(N(t),\mathcal{F}_t);\,t\geq 0\right\}$, counting the number of events that occurred up to time $t$, adapted to filtration
$\{\mathcal{F}_t\}$
\cite[Chapter 2]{aalen2008SurvivalEventHistory}.
Under mild assumptions, it has a Doob-Meyer decomposition
\begin{equation}
  \label{eq:cpm}
  N(t) = M(t)+\int_0^t \lambda(t)dt\,,
\end{equation}
where $\lambda(t)$ is a predictable process with respect to $\{\mathcal{F}_t\}$ called the intensity function,  
and $M(t)$ is a martingale, called the compensated counting process. 
Figure~\ref{fig:martingaleDecomp} shows an illustration of this decomposition for a Poisson process with sinusoidal intensity.
\input{fig0_martingaleDecompIllus.tex}

Consider now an empirical coupling measure $c$ between a (real or complex) predictable process $x(t)$ 
and $N(t)$ observed during time interval $[0,\,T]$, which takes the form of the stochastic integral (see e.g. \citet{protter2005stochastic})
\begin{equation}
  \label{eq:defC}
  \widehat{c} =\sum_{t_k<T} x(t_k) = \int_{0}^T x(t) dN(t)\,,
\end{equation}
where $\{t_k\}$ denote the jump times of the counting process 
(note that the $\rm PLV$ defined in Eq.(\ref{eq:plvDef}) is a normalized version of such coupling). 
The empirical coupling measure, $c$, can then be decomposed as
\begin{equation}\label{eq:doobmeyer}
\widehat{c} =  \int_{0}^T x(t) \lambda(t) dt + \int_{0}^T x(t) dM(t) \,.
\end{equation}
Interestingly, it can be shown that the second integral on the right hand-side is also a martingale  
(see \eg \citet[Theorem 18.7]{liptser2013statistics2}). 

In order to keep our results concise, 
we assume the following deterministic setting in the remainder of this paper 
(see section~\ref{sec:discussion} for potential extensions).
\begin{assum}\label{assum:deter}
  Assume the intensity function, $\lambda(t)=\lambda(t|\mathcal{F}_t)$ of $N(t)$, and the signal $x(t)$ are deterministic bounded left-continuous and adapted to $\mathcal{F}_t$ over $[0,\,T]$.
\end{assum}
Note this entails that $N(t)$ is a (possibly inhomogeneous) Poisson process 
\cite[Theorem 18.10]{liptser2013statistics2}. Under Assumption~\ref{assum:deter}, the terms of Eq.~(\ref{eq:doobmeyer}) separates the deterministic part from the (zero-mean) random fluctuations of the measure, that are integrally due to the martingale term. Using martingale properties, the statistics of the coupling measure are\footnote{See Appendix~\ref{app:Hmartingale} for more details.}
\begin{equation}\label{eq:uniCstats}
  c^* \triangleq \mathbb{E} \left[\widehat{c}\right] \!= \!\!\int_{0}^T \!\!\!\! x(t) \lambda(t) dt\quad \text{and}\quad \text{Var}[\widehat{c}] \!\! = \!\! 
  \mathbb{ E}\left[|\widehat{c}-c^*|^2\right]\!=\!\int_{0}^T \!\!\!\! |x|^2(t) \lambda(t) dt\,.
\end{equation}
In case $x(t)$ integrates to zero,  the \emph{expected coupling} $c^*$ reflects the covariation across time between $x(t)$ and the intensity of the point process up to random fluctuations.


%% file: fig0_martingaleDecompIllus.tex
\begin{figure}
  \centering
  \includegraphics[width=.6\linewidth]{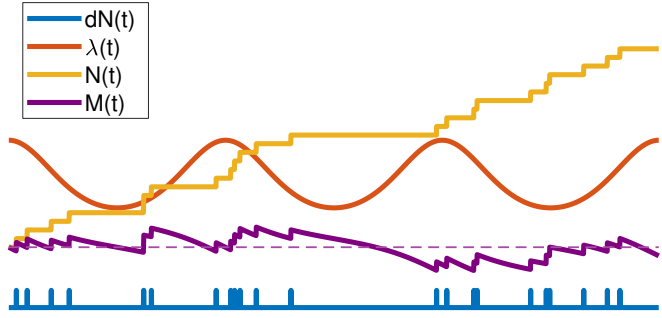}
  \caption{
    Doob-Meyer decomposition for an example inhomogenous Poisson process with oscillatory of rate $\lambda(t)$ of frequency $f = 1$ Hz, average firing rate $\lambda_0 = 5$ Hz (dashed line indicates the reference 0). See  section \ref{ssec:simulation_1d} for the detail of the simulation.
\label{fig:martingaleDecomp}
  }
  
\end{figure}


%% file: neuralComp_background_rmt.tex
As datasets get increasingly high dimensional, 
it becomes important to replace the above univariate measure $\widehat{c}$ 
by a quantity that summarizes the coupling between a large number of units and continuous signals. 
This extention leads to assessing the spectral properties of a coupling matrix $\widehat{\textbf{C}}$ 
which gathers all pairwise measurements. 
However, such task is non-trivial due to the martingale fluctuations affecting $\widehat{\textbf{C}}$, 
leading to spurious non-zero coupling coefficients and can also hide the deterministic structure of the matrix associated to significant coupling.

Random matrix theory allows investigating the spectral properties of some matrices in noisy settings by studying their asymptotic spectral properties as dimensions grows to infinity. 
Any $(p\times p)$ complex Hermitian or real symmetric matrix $\boldsymbol{M}$ has a set of $p$ real eigenvalues $\{\ell_k\}$ 
(where we put several times the same eigenvalue in the set according to its multiplicity). 
One classically studied quantity is then the \textit{empirical spectral distribution} (ESD) 
(or \textit{empirical eigenvalue distribution}, 
see \eg \citet{mingo2017free,anderson2010introduction}) of the set of all eigenvalues $\{\ell_k\}$. 
ESD indistinctly refers
(with a slight abuse of language), 
to either the probability measure 
(also called \textit{spectral measure} in our case)
\[
\mu_{\boldsymbol{M}} (t) = \frac{1}{p}(\delta_{\ell_1}(t)+\cdots+\delta_{\ell_p}(t)),\, t\in \mathbb{R}\,,
\]
where $\delta_{\ell_k}$ is the dirac measure with unit mass in $\ell_k$,
or to its associated cumulative distribution
\[
F_{\boldsymbol{M}}(t) = \int_{-\infty}^{t} d\mu_M(s)\,.
\]
Seminal works by \citet{wigner55,wigner1958distribution}, \citet{marchenko1967distribution} and many others have established the convergence 
of the ESD of large random matrix ensembles (see Appendix~\ref{app:convergence} for the precise notions of convergence). 
In particular, for a sequence of matrices $\left\{\boldsymbol{X}_n\right\}_{n>0}$ of dimension $p\times n$ such that
 $\frac{p}{n}\underset{n\rightarrow+\infty}{\rightarrow} \alpha\leq 1$, 
with coefficients sampled i.i.d. from a (possibly complex) standard Normal distribution,
the ESD of the Wishart matrix $\textbf{S}_n=\frac{1}{n}\boldsymbol{X}_n \boldsymbol{X}_n^H$
(where $.^H$ indicates the transposed complex conjugate)
converges to the Marchenko-Pastur (MP) law $\mu_{MP}(x)$ \citep{marchenko1967distribution} with density
\begin{equation}
  \label{eq:mpLaw}
  \frac{d\mu_{MP}}{dx}(x) = \begin{cases}
    \frac{1}{2\pi \alpha x}\sqrt{(b-x)(x-a)}\,,&a\leq x\leq b,\\
    0\,,&\text{otherwise\,,}
  \end{cases}
\end{equation}
with $a=(1-\sqrt{\alpha})^2$ and $b=(1+\sqrt{\alpha})^2$. Additionally, the smallest and largest eigenvectors converge to $a$ and $b$, respectively. Importantly, these convergences also hold in the case $\alpha>1$, but Eq.(\ref{eq:mpLaw}) is modified to account for the rank deficiency of the Wishart matrix, imposing $p-n$ zero eigenvalues in the spectrum (see Section~\ref{sec:wishart} for details).

We will show that the martingale fluctuations of the coupling matrices also cause spectral convergence to the MP law, in absence of actual coupling between the signals. Recent results on the low rank Perturbation \citep{capitaine2016SpectrumDeformedRandom,loubaton2011almost,Benaych2012singular} of random matrices 
suggest we can exploit this convergence to further assess the significance of eigenvalues of the coupling matrix with respect to those purely resulting from random fluctuations.


%% file: uniVar_math.tex

  We consider the setting of $K\geq 1$ independent trials of measurements on $[0,\,T]$ available
to estimate the coupling statistics by the trial average  
\[\widehat{\text{\rm c}}_K=\frac{1}{K}\sum_{k = 1}^K\int_0^T x(t)dN^{(k)}(t)\,,  
\]
where $\{N^{(k)}\}$ are $K$ independent copies of the process $N(t)$, associated to each trial. As this paper focuses on the statistical properties induced by the intrinsic variability of point process realizations, we assumed above that the continuous signal that does not change across trials. However, including some forms of variability across trials, such as random time shifts affecting all processes in the same way, would not affect the results, baring additional technical details.
 
We exploit a Central Limit Theorem (CLT) for martingales to show the residual variability 
(difference between empirically estimated $\widehat{\text{\rm c}}_K$ and the expected coupling $\text{\rm c}^*$ of Eq.(\ref{eq:uniCstats})) 
is asymptotically normally distributed.
We formally state it in Theorem~\ref{thm:1Dcoupling}.

\begin{thm}
  \label{thm:1Dcoupling}
  Assume $(\mathcal{F}_t,x(t),\lambda(t))$ satisfy Assumption~\ref{assum:deter}. Then,
  \[
    \mathbb{E}[\widehat{\text{\rm c}}_K]  \triangleq\text{\rm c}^* =  \int_0^T x(t)\lambda(t)dt \quad \text{and} \quad
    {\rm Var}[\widehat{\text{\rm c}}_K] = \frac{1}{K}\int_0^T x^2(t)\lambda(t)dt\,.
\]
Moreover, as the number of trials increases, fluctuations converge in distribution
\[    \sqrt{K}\left(\widehat{\text{\rm c}}_K-\text{\rm c}^*\right)\underset{K\rightarrow +\infty
    }{\longrightarrow}\mathcal{N}\left(0,\int_0^T x^2(t)\lambda(t)dt\right)\,.
  \]
\end{thm}
\begin{proof}[Sketch of the proof]
  It relies on the decomposition of Eq.(\ref{eq:doobmeyer}). 
  As described in Appendix~\ref{app:Hmartingale}, 
  the martingale property is preserved by the stochastic integral term, 
  and allows us to exploit a martingale CLT to prove convergence to a Gaussian distribution.
\end{proof}
We can exploit Theorem~\ref{thm:1Dcoupling} 
to derive the asymptotic properties of the  $\rm PLV$ introduced in Section~\ref{sec:spc}.
For that, we adapt the empirical estimate of Eq.(\ref{eq:plvDef}) to the  $K$ trials setting introduced above and define
\begin{equation}\label{eq:multiTrialPLV}
  \widehat{\text{\rm PLV}}_K
  \!=\! \frac{1}{\sum_{k=1}^K N_k}\sum_{k=1}^K \sum_{j=1}^{N_k} e^{\boldsymbol{i} \phi(t_j^k)} \,,
\end{equation}
where $N_k$ is the number of events observed during trial $k$ and $\left\{t_j^k\right\}$ is the collection of the time stamps of these events. 
The specificity of this multi-trial estimate is to use a single normalization constant corresponding to the the total number of events pooled across trials
\footnote{This allows the normalization factor to converge to a deterministic quantity as $K\rightarrow+\infty$}.
For this estimate we get the following result.
\begin{corol}\label{corol:1DPLV}
  Assume $(\mathcal{F}_t,x(t)\!=\!e^{\boldsymbol{i}\phi(t)},\lambda(t))$ satisfy Assumption~\ref{assum:deter}, 
  where $\phi$ is real-valued and stands for the phase of the signal $x$. 
  Then the expectation of the $\rm PLV$ statistics $\widehat{\text{\rm PLV}}_K$ 
  estimated from $K$ trials of measurements on $[0,\,T]$ tends to the limit
  \begin{equation}
    \label{thePLV}
    \text{\rm PLV}^* \!=\! \int_0^T e^{\textbf{i}\phi(t)}\lambda(t)dt/\Lambda(T)\,,\quad 
    \text{with}\quad \Lambda(T)\!=\!\int_0^T \lambda(t)dt\,,
  \end{equation}
  Moreover, as $K\rightarrow +\infty$ the  residual,
  \begin{equation}\label{eq:residPLV}
    {\sqrt{K}}\left(\widehat{\text{\rm PLV}}_K-\text{\rm PLV}^*\right)\,,
  \end{equation}
  converges in distribution to a zero-mean complex Gaussian variable $Z$ 
  (\ie the joint distribution of real and imaginary parts is Gaussian), such that
  \[
    \text{\rm Cov}\!
    \left[
      \begin{matrix}
	\text{Re}\{Z\}\\
	\text{Im}\{Z\}
      \end{matrix}
    \right]\!=\!\frac{1}{\Lambda(T)^2}\int_0^T \!\!\! M(t)\lambda(t)dt\,,
    \,
    \text{ where}
    \,
    M(t)\!=\!\left[
      \begin{matrix}
	\cos^2(\phi(t))&\sin(2\phi(t))/2\\
	\sin(2\phi(t))/2&	\sin^2(\phi(t))
      \end{matrix}\right].
  \]
\end{corol}

\begin{proof}[Sketch of the proof]
  This relies on applying Theorem~\ref{thm:1Dcoupling}
  to the real and imaginary parts of $e^{\boldsymbol{i} \phi(t)}$.
  In addition, the coupling between both quantities is taken into account by replacing the variance of univariate quantities $\widetilde{V}(t)$ in
  Theorem~\ref{thm:1Dcoupling} by a covariance
  matrix that can be assessed with martingale results given in Appendix~\ref{app:Hmartingale}.
\end{proof}
\begin{remk}\label{remk:sinCoupl}
	For the simple case of a $T/k$-periodic sinusoidal signal ($k$ integer), such that $\phi(t)=2\pi kt/T$, and a sinusoidal modulation of the intensity with phase shift $\varphi_0$ and modulation amplitude $\varkappa$ such that
	\[
	\lambda(t)=\lambda_0\left(1+ \varkappa \cos\left( \phi(t)-\varphi_0 \right)\right),\, \lambda_0>0,\,0\leq\varkappa\leq 1\,,
	\]
	we get easily with trigonometric identities that $\text{\rm PLV}^*=\frac{1}{2}\varkappa e^{\boldsymbol{i}\varphi_0}$ and the residual of Eq.(\ref{eq:residPLV}) converges to an isotropic complex Gaussian of total variance\footnote{The sum of the variances of real and imaginary parts.} $\frac{1}{\lambda_0 T}$. Such that the coupling strength $\varkappa$ affects the mean but not the variance of the PLV estimate. 
	
	Also, it is easy to see that if $\lambda(t)$ is modulated by a sine wave at a different integer multiple $m\neq k$ of the fundamental frequency $1/T$, such that $\lambda(t)=\lambda_0 + \varkappa \cos\left( 2\pi mt/T-\varphi_0 \right)$, the $\text{\rm PLV}^*$ vanishes and the residual's variance remains the same. These properties make PLV straighforward to interpret and test for sinusoidal coupling with a carefully chosen observation duration $T$. The reader can refer to  Assumption~\ref{assum:linearPhase} and  Corollary~\ref{corol:uniuncoupl_linearPhase} for a formal statement of this remark.
\end{remk}
We can use Corollary~\ref{corol:1DPLV}
to predict the statistics of PLV estimates for other models of phase-locked spike trains. A classical model uses the \textit{von Mises distribution} (also known as circular normal distribution) 
with parameter $\kappa\geq 0$ to model the concentration of spiking probability around a specified locking phase $\phi_0$
(for more details see \citet{ashidaProcessingPhaseLockedSpikes2010}). 
The original model uses a purely sinusoidal time series by assuming a linearly increasing phase $\phi(t)=2\pi f t $, where $f$ is the modulating frequency, to derive the intensity of an inhomogeneous Poisson spike train
\begin{equation}\label{eq:VmrateClass}   
  \lambda(t)=\lambda_0 \exp \left( \kappa\cos(\phi(t)-\varphi_0) \right) \,.
\end{equation}
resulting in an analytical expression for the asymptotic complex-valued PLV,
\[
  \text{\rm PLV}^* = 
  e^{\textbf{i}\varphi_0}\frac{\int_{0}^{\pi} \cos(\theta) 
    \exp(\kappa\cos(\theta))d\theta}{\int_{0}^{\pi} 
    \exp(\kappa\cos(\theta))d\theta}=e^{\textbf{i}\varphi_0}\frac{I_1(\kappa)}{I_0(\kappa)}\,,
\]
with the $I_k$'s denoting the modified Bessel functions of the first kind for $k$ integer 
(see \eg \citet[p.~376]{abramowitz1972handbook}):
\[
  I_k(\kappa)=\frac{1}{\pi}\int_{0}^{\pi} \cos(k\theta)  \exp(\kappa\cos(\theta))d\theta\,.
\]
Compared to the sinusoidal coupling described in Remark~\ref{remk:sinCoupl}, whose PLV magnitude can reach at most 1/2, this model can achieve arbitrary large PLV, which might explain why it is more frequently used in applications.

The following corollaries
(Corollary~\ref{corol:kappaplv} and Corollary~\ref{corol:uniuncoupl})
derive the asymptotic covariance of the variability of the $\rm PLV$ estimate around this theoretical value 
(which is novel to the best of our knowledge). 
Furthermore, the results are derived in a more general model setting accounting for ``biases''\footnote{They are biases in the sense that one would expect a coupling measure to vanish if there is no coupling in the data generating procedure.} due to non-linear phase increases $\phi(t)$,
and observation intervals that are not multiples of the modulating oscillation period.
It should be noted that the mentioned biases are inherent to the estimator's definition. 
They happen independently of additional biases originating from the phase estimation procedure
(\eg phase extraction via Hilbert transform, see \citet{kovach2017BiasedLookPhase}).

We thus assume a coupling, parameterized by $\kappa$ between a possibly non-linearly increasing phase $\phi(t)$ and a point process with intensity
\begin{equation}\label{eq:Vmrate}          
  \lambda(t)=\lambda_0 \exp \left( \kappa\cos(\phi(t)-\varphi_0) \right)
  \frac{d\phi}{dt}(t)\,.
\end{equation}
Note that for linearly increasing phases, this coupling amounts to the classical von Mises model of Eq.(\ref{eq:VmrateClass}). The additional factor 
$\frac{d\phi}{dt}(t)$ allows to preserve the analytical expression of $\rm PLV$ statistics even for non-linearly increasing phases, providing a novel generalization of the von Mises model
(see Corollary~\ref{corol:kappaplv_linearPhase} in Appendix~\ref{app:addColos} for a simplified version of Corollary~\ref{corol:kappaplv} assuming a linearly increasing phase $\phi(t)=2\pi f t $).

\begin{corol}
  \label{corol:kappaplv}
  Under the assumptions of Corollary~\ref{corol:1DPLV}, 
  assume additionally that $\phi(t)$ is continuous, 
  strictly increasing and piece-wise differentiable on $[0,\,T]$ 
  and the intensity of the point-process is given by Eq.~(\ref{eq:Vmrate}), 
  for a given $\kappa \geq 0$,
  then the expectation of the multi-trial {\rm PLV} estimate converges 
  (for $K\rightarrow +\infty$) to
  \begin{equation}\label{eq:PLVmeanGeneral}
    \text{\rm PLV}^*=\frac{\int_{\phi(0)}^{\phi(T)} e^{\textbf{i}\theta}
      \exp(\kappa\cos(\theta-\varphi_0))d\theta}{\int_{\phi(0)}^{\phi(T)}
      \exp(\kappa\cos(\theta-\varphi_0))d\theta}\,.
  \end{equation}
  If in addition $[0,\,T]$ corresponds to an integer number of periods of the oscillation,
  \begin{equation}\label{eq:classicPLVbessel}
    \text{\rm PLV}^* =	e^{\textbf{i}\varphi_0}\frac{\int_{0}^{\pi} \cos(\theta)
      \exp(\kappa\cos(\theta))d\theta}{\int_{0}^{\pi} \exp(\kappa\cos(\theta))d\theta}
    = e^{\textbf{i}\varphi_0}\frac{I_1(\kappa)}{I_0(\kappa)}\,,
  \end{equation}    
  and the scaled residual
  $\sqrt{K}\left(\widehat{\text{\rm PLV}}_K - \text{\rm PLV}^* \right)$
  converges to a zero mean complex Gaussian $Z$ with the following covariance:
  \begin{equation}
    \label{eq:covPLSPT}
    \text{\rm Cov}
    \left[
      \begin{matrix}
        \text{Re}\{Ze^{-i\varphi_0}\}\\
        \text{Im}\{Ze^{-i\varphi_0}\}
      \end{matrix}
    \right]=\frac{1}{2\lambda_0 (\phi(T)\!-\!\phi(0)) I_0(\kappa)^2}\left[
      \begin{matrix}
        I_0(\kappa)\!+\!I_2(\kappa) & 0\\
        0 & I_0(\kappa)\!-\!I_2(\kappa)
      \end{matrix}\right]\,.
  \end{equation}
\end{corol}
\begin{proof}[Sketch of the proof]
  This is based on plugging the intensity function $\lambda(t)$ of Eq.(\ref{eq:Vmrate}) in Corollary~\ref{corol:1DPLV}.
  Using change of variable in the integrals ($\phi(t)$ to $\theta$) and exploiting the symmetries of the functions, 
  the integrals in the analytical expressions of the expectation and covariance turn into modified Bessel functions $I_k$ for $k$ integer.  
\end{proof}
The above result has important consequences on the assessment of $\rm PLV$ from data. 
In particular, it exhibits key experimental requirements for $\rm PLV$ estimates to match the classical Bessel functions expression of Eq.(\ref{eq:classicPLVbessel}). 
These are: 
(1) evaluate $\rm PLV$ on an integer number of periods 
(this is critical for trials with short duration), 
(2) take into account the fluctuations of the rate of increase of the phase $\phi(t)$ across the oscillation period. 
This second point is critical in applications where the phase is inferred from signals (such as LFPs) trough the Hilbert transform, 
as non-linearities of the underlying phenomena may lead to non-sinusoidal oscillations, 
with periodic fluctuations of the time derivative of the phase $\phi'(t)$. 
To further emphasize the consequences of this aspect, 
we also derive the asymptotic distribution of ${\rm PLV}$ for a
homogeneous Poisson process which corresponds to the special case $\kappa=0$ of the classical von Mises coupling of Eq.~(\ref{eq:Vmrate}). 
Although there is no actual coupling between events and the continuous signal in such case,\footnote{In the sense that we can generate the homogeneous spike train and the oscillation without parametric models that do not share any information} the
 non-linear phase increase leads asymptotically (for $K$ large) to a non-vanishing $\rm PLV$ estimate and to false detection of coupling.

\input{neuralComp_coro_uniuncoupl_v1.tex}


%% file: neuralComp_coro_uniuncoupl_v1.tex
\begin{corol}
  \label{corol:uniuncoupl}
  Under the assumptions of Corollary~\ref{corol:1DPLV},
  we assume additionally that the point process is homogeneous Poisson with rate $\lambda_0 $
  and that $\phi(t)$ is strictly increasing (almost everywhere) and differentiable on $[0,\,T]$. 
  Let $\theta \mapsto \tau(\theta)$ be its inverse function (such that $\tau(\phi(t))=t$).
  Then the expectation of $\widehat{\mbox{PLV}}_K$ converges (for $K\rightarrow +\infty$) to
  \begin{equation}
    \label{eq:theoPLVhomoSPT}
    \text{\rm PLV}^*=\frac{\int_{\phi(0)}^{\phi(T)} e^{\textbf{i}\theta}
      \tau'(\theta)d\theta}{\phi(T)-\phi(0)}\,,
  \end{equation}
  and the scaled residual,
  \[
    Z =  \sqrt{K}\left(\widehat{\rm PLV}_K - \rm PLV^* \right)\,,
  \]
  converges to a zero mean complex Gaussian:
  \begin{equation*}
    \sqrt{K}\left(\widehat{\rm PLV}_K - \rm PLV^* \right)
    \underset{K\rightarrow +\infty}
    {\longrightarrow}\mathcal{N}\left(
      \left[\begin{matrix}
          0\\
          0
        \end{matrix}\right], \rm Cov(Z)\right)\,,
  \end{equation*}
  with the following covariance
    \[
    \rm Cov(Z) = 
    \frac{1}{\lambda_0 T^2} \int_{\phi(0)}^{\phi(T)}
    \left[
      \begin{matrix}
	\cos^2(\theta)&\sin(2\theta)/2\\
	\sin(2\theta)/2&	\sin^2(\theta)
      \end{matrix}\right]
    \tau'(\theta) d\theta\,.
  \]
\end{corol}


\begin{proof}[Sketch of the proof]
  The result stems from using the intensity function $\lambda_0 $ in Corollary~\ref{corol:1DPLV}.
  Then using change of variable in the integrals and exploiting the symmetries of the functions.
\end{proof}

This corollary will be further illustrated in the next paragraphs.


%% file: uniVar_application_v1.tex
Corollary~\ref{corol:uniuncoupl} predicts scenarios
where in absence of modulation of spiking activity 
(having a constant intensity function $\lambda(t) = \lambda_0$) 
the expectation of the $\rm PLV$ estimates remain far from zero even when the number of trials is large
\ie the coupling between a homogeneous point process and a continuous
oscillatory signal would \emph{appear} significant and reflect a form of bias.
Corollary~\ref{corol:uniuncoupl} allows to compute this bias and therefore correct it.

One such case is when the observation interval is not an \emph{integer} number of oscillation periods.
To demonstrate it analytically, 
we can start from the $\rm PLV$ expectation with the constant intensity $\lambda_0$,
\begin{equation}
  \label{eq:plvOrgDef}
  {\rm PLV^*}
 =  \frac{\int_0^T e^{\boldsymbol{i}\phi(t)}\lambda(t)dt}{\int_0^T \lambda(t)dt}=\frac
    {\lambda_0 \int_{0}^{T} e^{\boldsymbol{i}\phi(t)} dt}
    {\lambda_0\int_{0}^{T} dt}=\frac{1}{T}\int_{0}^{T} e^{\boldsymbol{i}\phi(t)} dt
    \,.
 \end{equation}
Furthermore, we assume $\phi(t)$ has linear phase (Assumption~\ref{assum:linearPhase}):
$  \phi(t) =  2\pi f t$,
where $f$ is the frequency of oscillation of the continuous signal.
We then get:
\begin{align}
  {\rm PLV^*}
  & =\frac{1} {T} \int_{0}^{T} e^{\boldsymbol{i}2\pi f t} dt
 = \frac{1}{2\pi\alphaT\boldsymbol{i}} \left(e^{2\pi\alphaT\boldsymbol{i}} - 1\right)\,.
    \label{coro3thePred}
\end{align}
where $\alphaT =  Tf$
is the ratio of length of the time series ($T$) of signal to period of oscillation $\frac{1}{f}$.
As is noticeable in Eq.(\ref{coro3thePred}),
the coupling measure $\rm PLV^*$ is not zero when $\alphaT$ is not an integer number.
Notably, this bias affects both the magnitude and the phase of the $\rm PLV^*$ estimate.

Furthermore, even using an observation interval covering an integer number of periods,
non-linear increases in phase may lead to a non-vanishing $\rm PLV$.
This can be demonstrated with a simple example.
Again, we can start from original definition of $\rm PLV$ expectation (Eq.(\ref{thePLV})),
but now we do not assume the linearity of the phase. 
As introduced in Corollary~\ref{corol:uniuncoupl}, 
let $\theta \mapsto \tau(\theta)$ be the inverse of $\phi(t)$
and we use Eq.(\ref{eq:theoPLVhomoSPT}) to compute the $\rm PLV^*$.
We take a sinusoidal modulation over the oscillation period: 
$\tau(\theta)=\theta+\epsilon \sin(\theta)$ with $|\epsilon|<1$
\footnote{To guaranty the phase to be strictly increasing}. 
We thus get a non-vanishing asymptotic expected PLV:
\[
  \text{\rm PLV}^*=\frac{1}{2\pi}{\int_{0}^{2\pi} e^{\textbf{i}\theta} (1+\epsilon\cos(\theta))d\theta}=\epsilon{\int_{0}^{\pi} e^{\textbf{i}\theta} \cos(\theta)d\theta}=\epsilon/2\neq 0, \mbox{if } \epsilon\neq 0.
\]

Our theoretical framework can be used for developing methods to correct such biases.
In the linear phase setting, 
bias can be avoided simply by using an integer number of periods for coupling estimation.
Nevertheless, in the presence of non-linear phase evolution of the continuous signal, 
appropriate treatment is more challenging.
In this case, we can use the theoretical phase (if available) or its empirical estimate  to evaluate ${\rm PLV}^*$ under constant spike intensity assumptions with Eq.(\ref{eq:theoPLVhomoSPT})
and subtract this quantity to the estimated PLV.
For resolving issues that arise due to non-linearity of the estimated phase, 
specialized methods have been suggested.
For instance, \cite{hurtado2004phaseJump} dealt with phase jumps 
(which are a particular form of non-linearity) 
by interpolating the signal from the available data before and after the sudden change;
or \cite{cole2019cycle} introduce a \emph{cycle-by-cycle} method for analysis oscillatory dynamics.
In this method, they consider a linear phase for each detected cycle of oscillation.
Therefore, with this linear choice of phase, 
one can avoid the spurious coupling that can appear due to phase non-linearities.
Based on our framework, theoretically motivated methods that are not relying on the linearization of the phase can be developed.





%% file: neuralComp_experimetns_plvDist.tex
We demonstrate the outcome of our theoretical results using simulated phase-locked spike trains 
(similar to what has been introduced in Corollary~\ref{corol:kappaplv} and \ref{corol:kappaplv_linearPhase}) 
and sinusoidal oscillations.  
For generating phase-locked spike trains, 
we adopt the method introduced in \cite{ashidaProcessingPhaseLockedSpikes2010}. 
As the model has already been described elsewhere \citep{ashidaProcessingPhaseLockedSpikes2010} 
we restrict ourselves to a brief explanation.


To generate phase-locked or periodic spike trains based on the classical von Mises model
with rate $\lambda(t)$ as introduced in Eq.(\ref{eq:VmrateClass}), 
we use purely sinusoidal continuous signal $x(t)$ by assuming a linearly increasing phase $\phi(t)=2\pi f t $ with $f = 1 Hz$ and
various coupling strength ($\kappa$) 
(see Appendix~\ref{tb:parmasList} for lists of parameters used for each figure).
Based on this simulation we perform two numerical experiments 
to demonstrate the practical relevance of our (asymptotic) theoretical results.

\subsubsection*{Experiment 1}

In order to demonstrate the validity of Corollary~\ref{corol:kappaplv} and \ref{corol:kappaplv_linearPhase}, in Figure~\ref{fig:uniVarSim}
we show the empirical distribution of the normalized residual of the $\rm PLV$ estimate and compare it to its asymptotic theoretical distribution.  
We simulate two cases, one with homogeneous Poisson spike trains ($\kappa=0$) 
and one with phase-locked spike train ($\kappa = 0.5$) with Poisson statistics.
In both cases we observe the agreement between theory and simulation, 
as the joint distribution of real and imaginary part approaches an isotropic Gaussian. The slightly non-Gaussian shape of the real part histogram for $\kappa=0.5$ suggests however a slower convergence to the normal distribution in the case of coupled signals.

\input{fig4_uniVarDistAndIllust.tex}

\subsubsection*{Experiment 2}

We demonstrate an application of Corollary~\ref{corol:uniuncoupl} for bias evaluation with a simple simulation.
In section \ref{ssec:coro3theoPredict} was pointed out that  using a non-integer $fT$ ($T$ is not a multiple of the oscillation period)
can lead to spurious correlation between the point process and the oscillatory continuous signal.
By using Eq.(\ref{coro3thePred}) we can compute this bias.

We use a simulation similar to the one used in the previous experiment with an oscillatory signal and a homogeneous Poisson spike train ($\kappa = 0$) and investigate the coupling between these two signals.
If the length of the continuous signal is not an integer number of oscillation period
the $\rm PLV$ estimate has a non-zero empirical mean (see Figure~\ref{fig:coro3exp}A and B)
while when it is a multiple of number of oscillation period,
the estimate matches the ground truth (see Figure~\ref{fig:coro3exp}C).
In Figure~\ref{fig:coro3exp}D we compare the theoretical prediction and the numerical simulation for various length of the signals, showing this effect disappears when with an observation window covering a large number of oscillation periods.

\input{fig1_incompleteCycleBias.tex}


%% file: fig4_uniVarDistAndIllust.tex
\begin{figure}
  \centering
  \includegraphics[width=.8\linewidth]{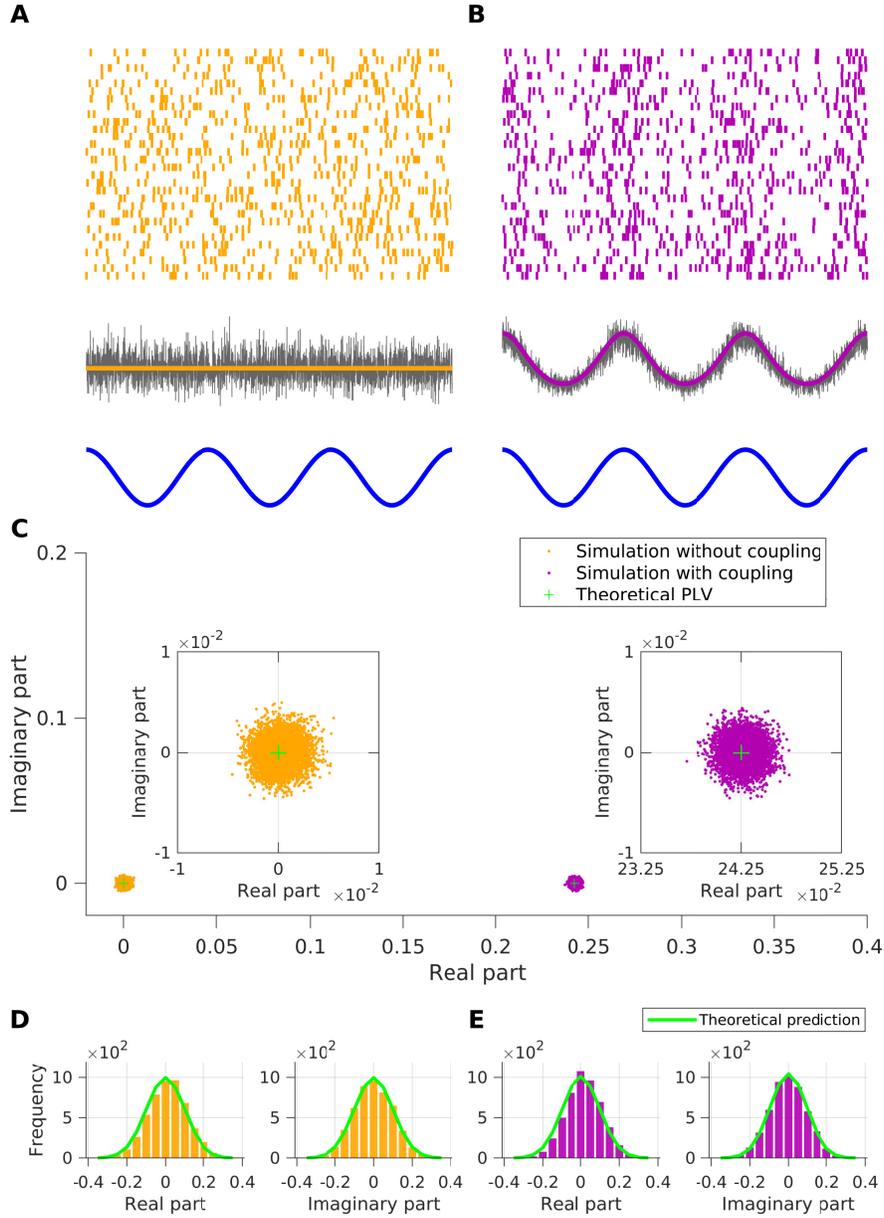}
  \caption{
    Simulation of
    (A) homogeneous Poisson spike trains and
    (B) phase-locked spike train  with Poisson statistics (von Mises model with $\kappa=0.5$).
    First row:
    example raster plot of the spikes.
    Second row:
    empirical firing rate (gray line) and
    ground truth firing rate (orange and purple trace).
    Third row:
    continuous signal $x(t)$.
    (C) Scatter plots represent the complex-valued $\rm PLV$s estimates.
    Each dot represents one realization of the simulation.
    Insets depict the zoomed version of both distributions.
    Green crosses indicate the theoretical complex-valued $\rm PLV$.
    (D-E)
    Histograms of real and imaginary parts
    for simulation (D) without coupling
    and (E) with coupling.
    Green lines indicate the theoretical predictions of corresponding distributions according to Corollary~\ref{corol:kappaplv} and \ref{corol:kappaplv_linearPhase}, 
    and the  bars indicate the empirical distributions.
    Note the subtle difference between real and imaginary part in (D) vs (E).
    See Table~\ref{table:figParam_uniVar} for parameters used for this figure. 
  }
  \label{fig:uniVarSim}
\end{figure}

%% file: fig1_incompleteCycleBias.tex
\begin{figure}[ht]
  \includegraphics[width=\linewidth]{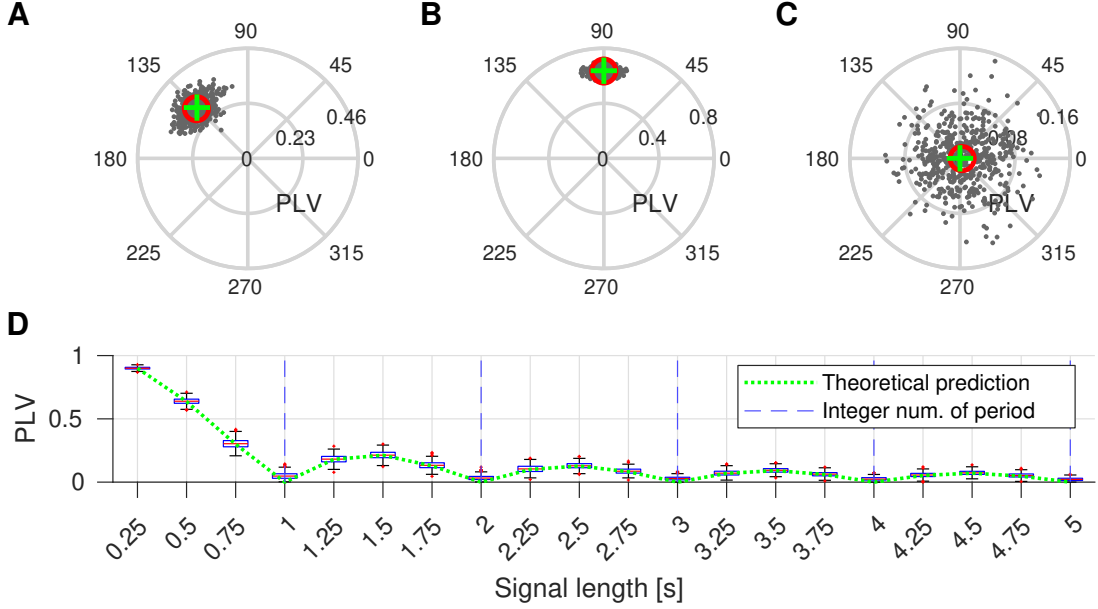}
  \caption{
    (A-C) Distribution of simulated complex-valued $\rm PLV$s (gray dots),
    average of the simulated $\rm PLV$s (red circle),
    and theoretical prediction based on Eq.(\ref{coro3thePred}) (green crosses)
    for (A) $\alphaT = 0.75$, (B) $\alphaT = 0.5$, and (C) $\alphaT = 1$.
    All complex-valued $\rm PLV$s are represented in the complex plane.
    Angles indicate the locking phase and the radius the $\rm PLV$.
    (D) $\rm PLV$ for different interval lengths $T$.
    Boxplots represent the simulated $\rm PLV$s and the dashed green trace
    represents theoretical prediction of the expectation based on Eq.(\ref{coro3thePred}). 
    Vertical broken blue lines indicates integer number of oscillation period.
    See Table~\ref{table:figParam_coro3exp} for parameters used for this figure.
    %
    \label{fig:coro3exp}
  }
\end{figure}


%% file: multVar_math.tex
\label{ssec:multiTheo}
We now replace Assumption~\ref{assum:deter} to adapt to this multivariate setting, 
restricting ourselves to the (null) hypothesis of no coupling between continuous signals and point process, 
reflected in an homogeneous Poisson process assumption. Let us denote $\bar{x}$ the complex conjugate of $x$ and $\delta$  the Kronecker delta symbol
\begin{equation}\label{eq:delta}
\delta_{lj} = \begin{cases}
1, \mbox{ if } l=j,\\ 0, \mbox{ otherwise.}
\end{cases}	
\end{equation}
\begin{assum}[Complex multivariate case]\label{assum:multi}
We consider an infinite sequence
$\{x_j(t)\}_{j\geq 1}$ 
of complex valued left-continuous
deterministic functions uniformly bounded on $[0,\,T]$ and assume 
\begin{itemize}
	\item[(1)] For all $i,j\geq 1$, 
$
	\frac{1}{T}\int_0^{T} \!\!\!\bar{x}_i x_j dt = \delta_{ij}\,\mbox{ and }\,\int_0^T \!\!\!x_i x_j dt=0\,.
$ 
\item[(2)] For all $i\geq 1$, $\int_0^T x_i dt=0$,
\item[(3)] There exist $0<\lambda_{\min}<\lambda_{\max}$ and a sequence of independent homogenous Poisson processes $\{N_i\}_{i\in \mathbb{N}^*}$'s with associated rates $\{\lambda_i\}_{i\in \mathbb{N}^*}$ in the interval $[\lambda_{\min},\,\lambda_{\max}]$.
\end{itemize}
\end{assum}
While the assumptions on $\{x_i(t)\}$ are designed for complex signals, which is the classical case when dealing with PLV-like quantities, the results of this section also hold for real signals by using assumption $\frac{1}{T}\int_0^{T} x_i x_j dt=\delta_{ij}$ instead of the above condition~(1). Condition~(2) is also added to ensure that there is no trivial bias leading to a non-vanishing expectation of the coupling coefficients (as noticed in a specific example in section~\ref{ssec:coro3theoPredict}). Indeed, when the time average of each signal vanishes, based on Theorem~\ref{thm:1Dcoupling}, the expectation of all univariate coupling measures for a homogeneous Poisson process vanish. We then exploit a multivariate generalization of the martingale CLT to characterize the distribution of the coupling matrix given these assumptions.
\begin{thm}
  \label{thm:multiclt}
  For a given $n,\,p\geq 1$, and all $K\geq1$ we use sequences of signals defined in Assumption~\ref{assum:multi} to build multivariate continuous signal $\boldsymbol{x}(t)=(x_j)_{j=1 \dots p}$ and
  $K$ independent copies of multivariate Poisson process $\boldsymbol{N}(t)=(N_i)_{i=1 \dots n}$ with rate vector $\boldsymbol{ \lambda}=[\lambda_1, \dots ,\lambda_n]^\top$. 
  Then the normalized coupling matrix $\sqrt{K}\widehat{\boldsymbol{C}}_K \text{diag}(\sqrt{T\boldsymbol{\lambda}})^{-1}$ of Eq.~(\ref{eq:defEmpMultVarC}) 
  converges in distribution for $K\rightarrow +\infty$ to a  
  matrix with i.i.d. complex standard normal coefficients.
\end{thm}
\begin{proof}[Sketch of the proof]
  This essentially uses a generalization of the CLT to multivariate point processes described in \citep[Appendix~B]{aalen2008SurvivalEventHistory}. 
Based on the statistics of stochastic integrals presented in Appendix~\ref{app:Hmartingale}, 
assumptions on $\boldsymbol{x}$ entail vanishing correlation between all matrix coefficients and lead to the analytical expression of the covariance matrix.
\end{proof}

This result suggests that for large $n$ and $p=p(n)$, coupling matrices $\widehat{\boldsymbol{C}}^n_K$ of increasing size can be used to build the Wishart-like matrix sequence
\begin{equation}\label{eq:wishartC}
\boldsymbol{S}_n \triangleq \frac{K}{n}\widehat{\boldsymbol{C}}^n_K
\text{diag}(T\boldsymbol{\lambda})^{-1}
(\widehat{\boldsymbol{C}}^n_K)^H
\end{equation}
whose ESD may converge to the \mpl. This is however not guaranteed by classical results due to the non-Gaussianity and dependence of the matrix coefficients of $\widehat{\boldsymbol{C}}^n_K$ for fixed $n$ and $K$. Convergence will thus depend on how much the departure from these assumptions plays a role as $n$ becomes large.
We show in the following Theorem that increasing the number of trials as a function of the dimension guaranties convergence to the MP law.
\begin{thm}\label{thm:cvMP}
In addition to Assumption~\ref{assum:multi}, 
assume an increasing, positive integer sequences $\{p(n),K(n)\}_{n\in \mathbb{N}^*}$ such that $\frac{p(n)}{n}\underset{n\rightarrow+\infty}{\longrightarrow}\alpha\in(0,+\infty)$, and
\begin{equation}\label{eq:4momCond}
\frac{1}{n^2 K(n)^2}\sum_{\Gamma}\left(\int_0^T\bar{x}_{j}{x}_{l}{x}_{j'}\bar{x}_{l'}dt\right)^2\rightarrow 0\,, \mbox{ uniformly in } k\leq n\,,
\end{equation}
where $\Gamma \!=\!\{(j,l,j',l'):1\leq j,l,j',l'\leq p\}\!\setminus\! \{(j,l,j',l'): j\!=\!j'\!\neq\! l\!=\!l' \mbox{ or }j\!=\!l'\!\neq\! j'\!=\!l\}$.
Consider the sequence $\{\widehat{\boldsymbol{C}}_{K(n)}^n\}_{n\in \mathbb{N}^*}$ built as in Theorem~\ref{thm:multiclt} for $p=p(n)$,
then the corresponding sequence $\{\boldsymbol{S}_n\}$ defined by Eq.(\ref{eq:wishartC})
has an ESD converging weakly with probability one to the MP law of Eq.(\ref{eq:mpLaw}).
\end{thm}

\begin{proof}[Sketch of the proof]
  We use Theorem 1.1 of \citet{bai2008large} addressing the case of matrices with dependence of coefficients within columns.
  We use It\^o's formula (see Appendix~\ref{app:addBackground}) to check the simplified necessary conditions provided in Corollary~1.1 of \citet{bai2008large}.
  This implies convergence of the Stieltjes transform to the same function as the transform of the MP distribution. 
  By classical results on the Stieltjes transform
  \cite[Theorem~2.4.4]{anderson2010introduction}, 
  this implies  weak convergence to the MP measure
  (\ie convergence for the weak topology --- see Appendix \ref{app:convergence}).
\end{proof}

\begin{remk}
	Condition in Eq.(\ref{eq:4momCond}) determines how many trials are needed at most for spectral convergence. Due to the uniform boundedness assuption on signal $\boldsymbol{x}(t)$, and given the number of terms in the sum is bounded by $n^4$, we can already see that $\frac{n}{K(n)}\rightarrow 0$, i.e. having number of trials increasing at a even slightly faster rate than dimension,  is enough for convergence for any choice of continuous signals respecting orthonormality Assumption~\ref{assum:multi}(1). However, there are cases where even less trials than dimensions are required. An important example is the Fourier basis of the $[0,\,T]$ interval, $\boldsymbol{x}_l(t)=\exp(\boldsymbol{i}2\pi l t/T)$. Then all terms in the sum of Eq.(\ref{eq:4momCond}) vanish, except the ones satisfying $j-j'-l+l'=0$, such that we are left with a number of bounded terms that scale with $n^3$, as a consequence, the condition on the number of trials to achieve spectral convergence becomes $\frac{\sqrt{n}}{K(n)}\rightarrow 0$, such that we need increasingly less trials than dimensions.
\end{remk}

This convergence of the spectral measure to the MP law guaranties eigenvalues do not accumulate in a large proportion above the upper-end of the support of the MP law, however, they do no provide rigorous guaranties regarding convergence of  individual eigenvalues, and in particular, the largest eigenvalue. Although such convergence is satisfied in classical settings (Gaussian i.i.d coefficients), they typically require stronger assumptions than for the (weak) spectral convergence to the MP law, and still only very few results are available in the non-i.i.d. setting. We could however prove such convergence by adding a constraint to our model.
\begin{thm}\label{thm:evCvEnds}
In addition to Assumption~\ref{assum:multi}, assume all homogeneous rates $\lambda_k$ are equal.
Assume two increasing, positive integer sequences $\{p(n),K(n)\}_{n\in \mathbb{N}^*}$ such that 
\begin{equation}\label{eq:cvEVcondition}
\textstyle
\frac{p(n)}{n}{\rightarrow}\alpha\in(0,+\infty)\,\quad \mbox{and}\,\quad \frac{1}{K(n)}\sum_{1\leq i,k\leq p(n)}\int_{0}^{T}|x_i x_j|^2(t)dt<B\,,
\end{equation} for some constant $B$.
Then for the sequence $\{\widehat{\boldsymbol{C}}_{K(n)}^n\}_{n\in \mathbb{N}^*}$ built in Theorem~\ref{thm:multiclt} for $p=p(n)$, the corresponding sequence $\{\boldsymbol{S}_n\}$ defined by Eq.(\ref{eq:wishartC})
has an ESD converging weakly with probability one to the MP law of Eq.(\ref{eq:mpLaw}). Moreover, Let $\ell_1$ and $\ell_{p}$ the largest and smallest eigenvalues of $\{\boldsymbol{S}_n\}$, respectively, then in probability
\[
\ell_1(n)\rightarrow (1+\sqrt{\alpha})^2\quad\mbox{and}\quad\ell_p(n)\rightarrow (1-\sqrt{\alpha})^2\boldsymbol{1}_{\alpha<1}\,. 
\]	
\end{thm}
\begin{proof}[Sketch of the proof]
	The identical intensities allows us to use the result of \citep{chafai2018convergence} for matrices with i.i.d. columns. We first checked their proof holds also for the complex case by replacing symetric matrices by Hermitian matrices, and squared scalar product by absolute squared hermitian product.  We satisfy their Strong Tail Projection (STP) assumption using Chebyshev's inequality. The necessary fourth order moment conditions exploit the same stochastic integration results as Theorem~\ref{thm:cvMP}.
\end{proof}
\begin{remk}
	Without additional assumptions, the moment condition of Eq.(\ref{eq:cvEVcondition}) is satisfied by choosing $K(n)=n^2$ (as there are $p^2$ bounded moments, scaling as $n^2$ when $n$ grows).
	It is likely from the proof that taking into account more information about the moments of the continuous signal sequence $\{x_j\}$, we can achieve convergence with a lower rate of increase for the number of trials. This is left to future work.
\end{remk}
This result thus provides the guaranties  that under a null hypothesis of no coupling (due to homogeneity of the Poisson processes), the extreme eigenvalues of $S$ will asymptotically cover exactly the full support of the MP law. This will be used in section~\ref{sec:signifEV} to assess significance of the eigenvalues $\ell_k$ by simply checking whether the are larger than $(1+\sqrt{\alpha})^2$.

This significance analysis relies as well on understanding what happens to the eigenvalues when the model departs from the null hypothesis. In a practical setting, we hypothesize that the coupling matrix has a deterministic structure superimposed to the martingale noise modeled in the above results. One qualitative justification of this assumption can be found in Remark~\ref{remk:sinCoupl},  showing that for sinusoidal coupling, 
an non-vanishing expectation proportional to the coupling is superimposed to martingale noise whose distribution is unaffected by coupling, such that the noisy part of the matrix satisfies the conditions of the above theorems. 
As typically done in applications, 
we are mostly interested in low rank structure controlled by the largest singular values of the coupling matrix, providing an interpretable summary of the multivariate interactions.

This naturally leads to modeling departure from the null hypothesis with a low rank perturbation assumption. 
In such case, we assume that the eigenvalue related to significant coupling appear in the spectrum of the perturbed matrix, and can be isolated from the remaining eigenvalues associated to the martingale noise. 
This intuition is justified by results in the case of the Wishart ensemble 
\citep{loubaton2011almost} (see also \citet{Benaych2012singular} for a more general result and \citet{capitaine2016SpectrumDeformedRandom} for an overview of matrix perturbation results), that we restate here:
\begin{thm}[From \cite{loubaton2011almost}, Theorem~6]
  \label{thm:loubaton}
  Let $\boldsymbol{X}_n$ be a $n\times p$ the sequence of i.i.d. complex Gaussian matrices defined in section \ref{sec:RMTbackground}, 
  and $A_n$ be a finite rank perturbation of the null matrix with non zero eigenvalues $\theta_i$. 
  Let $\boldsymbol{M}_n =
  (\frac{1}{\sqrt{n}}\boldsymbol{X}_n + A_n)(\frac{1}{\sqrt{n}}\boldsymbol{X}_n+A_n)^H$. 
  Then as $n \rightarrow \infty$ and $\frac{p}{n} \rightarrow \alpha \in (0,1)$, almost surely,
  \[
    \lambda_i(M_n)\rightarrow 
    \begin{cases}
      \frac{(1+\theta_i)(c+\theta_i)}{\theta_i},\, \mbox{if } \theta_i>\sqrt{\alpha},\\
      (1+\sqrt{\alpha})^2,\, \mbox{otherwise}\,.
    \end{cases}
  \]
\end{thm}
A demonstration that this further applies rigorously to our non-Gaussian, non-iid case is left to further work (but see \citet{Benaych2012singular} for a generalization in this direction).
This results shows the upper end of the MP support is indeed the critical threshold for the eigenvalues of $A_n$ to stand out from the noise. 
Below this threshold, the largest eigenvalue convergence to the upper end of the support of the MP distribution is not informative about $\theta_i$. 
Above this threshold, the value of $\theta_i$ can be recovered, 
and detected by comparing the largest eigenvalue to the upper end of the MP distribution.




We next illustrate the interest of these theoretical predictions
in the context of neural time series by reliably quantifying the interplay between multi channel LFP signals and the spiking of multiple neurons. Nevertheless, the results are potentially applicable in other domains as well.
In Neuroscience, $\boldsymbol{x}$ may represent LFP measurements collected on each recording channel, 
and $\boldsymbol{N}$ the spiking activity of different neurons, called units. 
The number of recording channels $n_c$ and recorded units $n_u$ correspond to $p$ and $n$ respectively. 
These number may differ, and as a consequence, the coupling matrix is generally rectangular.


%% file: neuralComp_appInSigcAssess.tex
In order to statistically assess the significance of the largest singular value(s) of 
coupling matrix $\widehat{\boldsymbol{C}}$ 
-- considered as a measure of coupling between point processes and continuous signals --
we need a null hypothesis.
Hypothesis testing based on generation of surrogate data is one of the common methods for significant assessment in Neuroscience and other fields.
Generating appropriate surrogate data can not only be challenging
(see \cite{grunDataDrivenSignificanceEstimation2009,elsayed2017structure} for examples in Neuroscience), 
but also computationally expensive due to increasingly large dimension of modern datasets.
Exploiting our theoretical results for this setting
allows us to perform such statistical assessment in a principled way, without using surrogate data, sparing computational resources.

In order to exploit the results of the theoretical part, it is best to preprocess the $p\times q$ matrix of time-discretized signals $\rm \textbf{L}$ that correspond to $q$ samples over interval $[0,\,T]$, with sampling interval $\Delta=T/q$. The chosen signals are driven by the application (in our case they are preprocessed LFPs, see section~\ref{ssec:simulation_multD} for a simulation reproducing the context of neurophysiolgy data). We assume the rows of $\rm \textbf{L}$ sum to zero to match Assumption~\ref{assum:multi}(2) (and avoid bias in the coupling measure similar to what is described in section~\ref{ssec:simulation_1d}). We then need to process further this signal such that Assumption~\ref{assum:multi}(1) is satisfied approximately. In order to achieve this, we perform classical whitening of the signals to generate matrix $\rm \textbf{X}$, the discrete time approximation of $\boldsymbol{x}(t)$, according to
\begin{equation}\label{eq:whiten}\displaystyle
{\rm \textbf{X} }= {\rm \textbf{W}\textbf{L}},\quad \mbox{ with } \quad {\rm \textbf{W}}=\left( \rm \frac{1}{q}\textbf{L}\textbf{L}^H\right)^{-\frac{1}{2}}\,,
\end{equation}
where the power in the expression of the whitening matrix ${\rm \textbf{W}}$ describes the inversion of a matrix square root, typically achieved via eigenvalue decomposition, and which may require PCA-like dimensionality reduction in practice to minimize the numerical effects of small eigenvalues. This procedure decorrelates the martingale fluctuations of coefficients within the same column of the coupling matrix (see Theorem~\ref{thm:multiclt}), a key requirement for convergence to the MP law.

%

As explained in section~\ref{ssec:multiTheo}, theoretical results support to use 
$
\theta_{DET} = (1+\sqrt{\alpha})^2\,.
$
 the upper end of the support of the MP law, as a detection threshold for the significance of the eigenvalues of the
hermitian matrix,
\begin{equation*}
\boldsymbol{S}_n=\frac{K}{n}\widehat{\boldsymbol{C}}^n_K
\text{diag}(T\boldsymbol{\lambda}_0)^{-1}
(\widehat{\boldsymbol{C}}^n_K)^H\,.
\end{equation*}

The null hypothesis of non-significance of the $k$-th largest singular value $\widehat{\sigma}_k$ of the normalized coupling matrix 
\begin{equation*}\label{eq:normCoupl}
\sqrt{K}\widehat{\boldsymbol{C}}^n_K
\text{diag}(\sqrt{T\boldsymbol{\lambda}_0})^{-1}
\end{equation*}
  should thus be rejected if,
the corresponding $k$-th largest eigenvalue $\ell_k$ of $\boldsymbol{S}_n$ is superior to the significance threshold, leading to the condition:
\begin{equation}
  \label{eq:gPLVsigThr}
  \widehat{\sigma}_k= \sqrt{n\ell_k} > \sqrt{n\theta_{DET}}=\sqrt{n}(1+\sqrt{\alpha})\,,
\end{equation}
and therefore, 
there is a significant coupling between the multivariate point process and continuous signal.
An illustration of our overall significance assessment approach is shown in Figure~\ref{fig:sigAssesDemo}.

\input{fig5_demoSigAsses.tex}

%% file: fig5_demoSigAsses.tex
\begin{figure}[hpt]
  \centering
  \includegraphics[width = \linewidth]{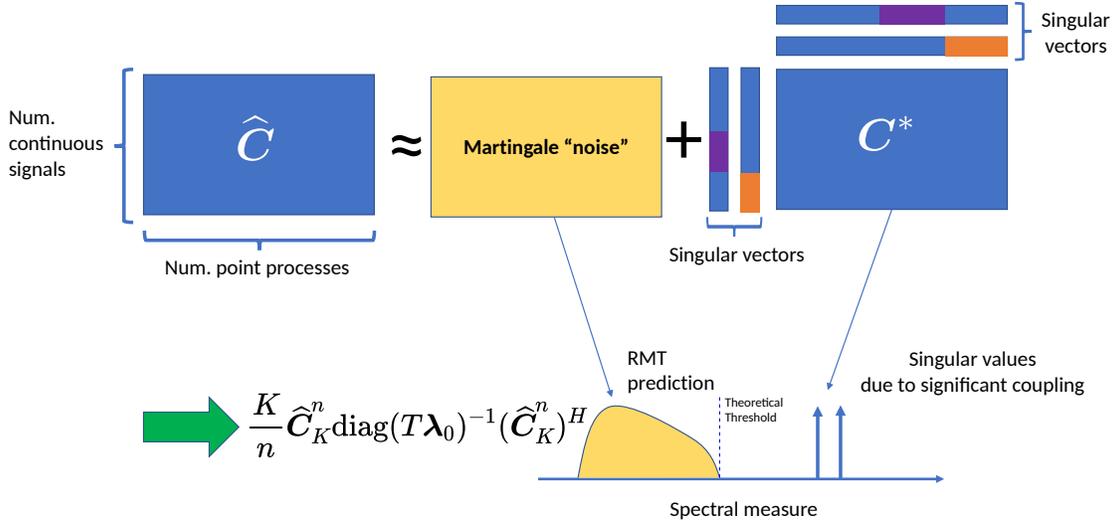}
  \caption{
    We assume $\widehat{\boldsymbol{C}}$ is a superposition of martingale noise and a low rank deterministic matrix $\boldsymbol{C}^*$ reflecting the actual coupling.  
    If the singular values of a normalized version of $\boldsymbol{C}^*$ are large enough (larger than the upper end of the MP law support), 
    theory suggests that they will
    correspond to the largest eigenvalues of $\boldsymbol{S}_n$ appearing beyond the support of MP distributed eigenvalues reflection martingale noise.
    They can thus be detected with a simple thresholding approach (see Eq.(\ref{eq:gPLVsigThr})).
    \label{fig:sigAssesDemo}
  }
\end{figure}

%% file: neuralComp_experimetns_multVar_v2.tex
\label{ssec:signif} We use a simulation to demonstrate the outcome of our (asymptotic) theoretical results on mutlivariate coupling.
Similar to the simulations of section~\ref{ssec:simulation_1d} for the univariate case,
we use simulated phase-locked spike trains with Poisson statistics. 
The main difference between this simulation and the previous one is in synthesizing the LFP. 
In order to simulate multi-channel oscillatory signals that lead to a low rank structure for $\boldsymbol{C}^*$
we use a combination of noisy oscillatory components. 

The LFPs contain $N_{osc}$ oscillatory groups of channels, each channel $l$ within the same group contains same oscillatory components with index $j(l)$, with the time course of all these components being
$O_j(t) = e^{2 \pi \boldsymbol{i} f_j t}$, $j \in \{1,\dots,N_{osc}\}$, with all frequencies $f_i$ comprised in
range $[f_{min},\,f_{max}]$, and all multiple of $1/T$. Due to necessary time axis discretization, the bracket notation $[t]$ indicates the oscillation is sampled at equispaced discrete times $t=\{k\Delta\}_{k=1, \dots, q}$. 
The synthesized discrete time multichannel LFP ($\boldsymbol{\Psi}[t] = \{\psi_l[t]\}_{l = 1, \dots, n_c}$)
can be written as
\begin{equation}
  \Psi_l[t] = O_{j(l)}[t] \odot \exp \left( \boldsymbol{i} \eta_l[t]\right)\,.
\end{equation}
with $\odot$ entrywise product and $\{\eta_l[t])\}$ i.i.d. sampled (white) phase noises contaminating each channel independently 
(see Appendix~\ref{app:circAddNoise} for more details).

In this simulation,
the frequency of the oscillatory component are ranging from 11-15 Hz. 
We used 100 LFP channels ($n_c = 100$) and different choices for the number of spiking units 
(10, 50, and 90).
Spiking activities are simulated in different scenarios, with and without coupling to the LFP oscillations.
In the latter case, 
we have 2 populations of neurons (each consisting in 1/5th of the total number of neurons) 
that are each coupled to one of the oscillatory groups of LFP channels.
Both populations are coupled to their respective oscillation with 
identical strength ($\kappa=0.15$) and phase ($\phi_0 = 0$).

To compute the coupling matrix $\widehat{\boldsymbol{C}}_K$, we first preprocess $\boldsymbol{\Psi}[t]$ by applying band-pass filtering in a range covering $[f_{\min},\,f_{\max} ]$, and convert it to an analytic signal via the discrete time Hilbert transform, leading to data matrix ${\rm \textbf{L}}$, following the standards of PLV analysis in Neuroscience \citep{Chavez2006}.

This signal matrix is then whitened according to Eq.(\ref{eq:whiten}) to yield matrix $\rm \textbf{X}$ the discrete time version of $\boldsymbol{x}(t)$.
The coupling matrix $\widehat{\boldsymbol{C}}_K$ is then computed according to Eq.(\ref{eq:defEmpMultVarC}) using 10 trials (barring trivial approximation to the closest time sample in $\rm \textbf{X}$).

Then in order to approximate the normalization $\sqrt{K}\widehat{\boldsymbol{C}}^n_K
\text{diag}(\sqrt{T\boldsymbol{\lambda}_0})^{-1}$ based on empirical data, we use the total number of events for unit $u$ occurring across all $K$ trials $N_{{\rm tot}}^u={\sum_{k = 1}^K N_k^u}$ , and multiply each column $u$ of the coupling matrix by 
\[
\frac{K}{\sqrt{N_{{\rm tot}}^u}}\approx \frac{K}{\sqrt{K\int_{0}^{T}\lambda_u(t)dt}}\,,
\] 
for the corresponding unit $u$, asymptotically matching the theoretical normalization in the homogeneous Poisson case.

We observe in Figure~\ref{fig:multVarSim}A that in the absence of coupling,
the distribution of eigenvalues originating from the random matrix structure is very close to the theoretically predicted MP  distribution,
and  Figure~\ref{fig:multVarSim}B where we have coupling between spike and 
eigenvalues reflecting the coupling  are  beyond the MP support 
(blue line in Figure~\ref{fig:multVarSim}),
and the eigenvalue bulk below the threshold is also close to MP distribution.
This suggest an easy thresholding approach for significance assessment.

\input{fig3_multVarDist.tex}


%% file: fig3_multVarDist.tex
\begin{figure}
  \centering
  \includegraphics[width=.85\linewidth]{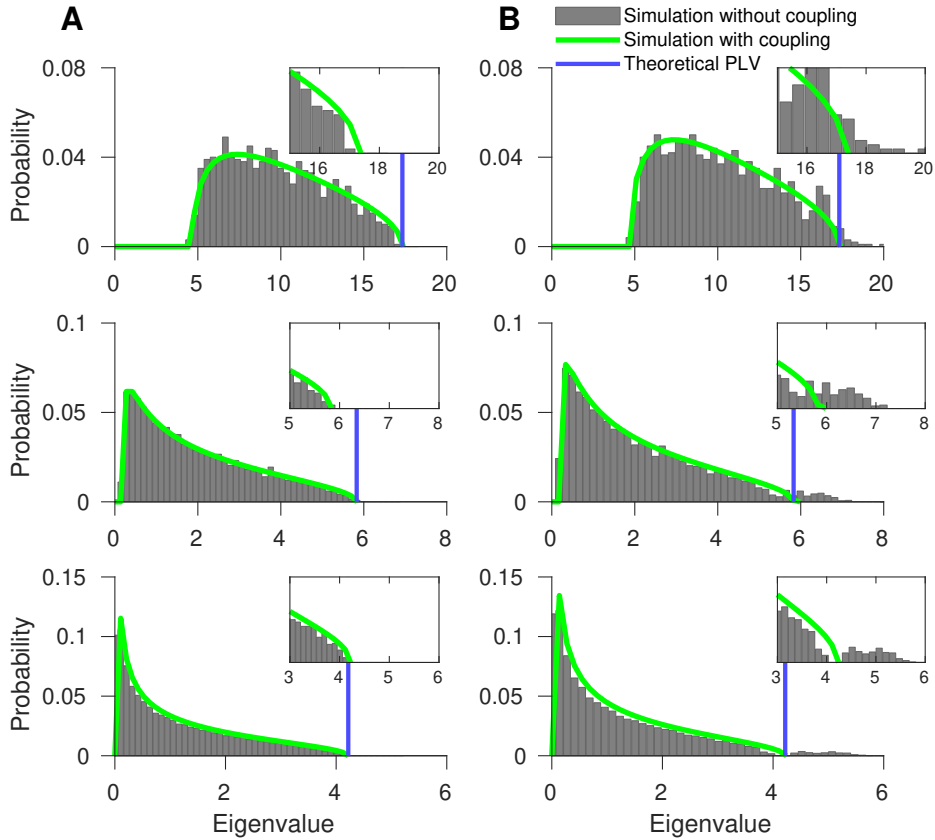}
  \caption{
    Theoretical Marchenko-Pastur distribution (green lines)
    and empirical distribution (gray bars)
    for (A) simulation without coupling ($\kappa=0$) and
    (B) with coupling  ($\kappa=0.15$) between multivariate spikes and LFP.
    Rows  are representing the spectral distribution of simulations
    with different number of spiking units,
    row 1, 2, and 3 respectively 10, 50, and 90
    (which leads to different $\alpha$ for MP law).
    Insets zoom the tail of the distributions.
    Parameters used for this figure are denoted in Tabel~\ref{table:figParam_multVarSim}
    \label{fig:multVarSim}
  }
\end{figure}


%% file: neuralComp_discussion.tex
\subsubsection*{Insights for data analysis}
Our theoretical results provide guaranties for specific coupling models to respect univariate and multivariate asymptotic statistics that can be easy exploited for statistical testing. The required assumptions provide guidelines for practical settings that are likely of interest beyond the strict framework that we imposed to get the rigorous results. For the univariate coupling measure, corollaries and simulations point out the importance of the choice of observation interval $[0,T]$, which is particularly sensitive when considering short intervals covering only few oscillation periods. This is the case when doing time resolved analysis or dealing with experiments with short trial duration. Moreover, the univariate results also emphasize the effect of non-linear phase increases, highly relevant in Neuroscience due to pervasive effects of non-linear dynamics in the mesoscopic signals. Our result provide asymptotic bias correction terms that can be used for statistical testing.

In the same way, theoretical results in the multivariate setting may seem to be constrained by our assumptions, but provide critical guidelines to interpret singular values. First, whitening the continuous signals and normalizing the coupling by the square root of the rate are key preprocessing steps to be able to make the asymptotic behavior of the martingale noise invariant to the specifics of the data at hand. This then reduces to an analytical model, the MP law, dependent only on a single matrix shape parameter. After assessment of the significance of the singular value of the normalized coupling matrix, it is of course possible to revert these preprocessing steps to get a low rank approximation of the original coupling matrix (non-normalized, non-whitened) to summarize significant coupling structure in an interpretable way.
A second insight provided by the multivariate results is the role of ``fourth order moments'' of the continuous signals, represented by the integrals of order four monomials of components of $\boldsymbol{x}(t)$, in the MP convergence results. The magnitude of these moments determines the amount of trials asymptotically needed to achieve convergence. Since these moments can be estimated empirically, we can check how they grow with the dimension of the signals in a specific application. With our minimal assumptions on the signals, the number of trials need only to grow at most sublinearly in the dimension for spectral convergence; however, we could only show that convergence of the largest eigenvalues requires at most quadratic increase in the dimension $n$. This last result might be improved in furture work, with extra assumptions, to reach linear growth.

Our theoretical results can be extended in two direction in future works.
The first is toward exploiting point processes different from inhomogeneous Poisson
(\eg Hawkes process) in order to be able to apply the framework in applications where the process intensities are stochastic.
The second direction is toward exploiting recent developments in Random Matrix Theory, 
in order to develop a probabilistic significance assessment. 

\input{neuralComp_potentialExtention.tex}


%% file: neuralComp_potentialExtention.tex
\subsubsection*{Extension of signal assumptions}
Our theoretical results assume deterministic continuous signals and point process intensities (see Assumption~\ref{assum:deter}). This entails limitations, such as implicitly assuming the considered point processes are (homogeneous or inhomogeneous) Poisson processes. This assumption may be too restrictive in realistic scenarios
(\eg see  \citet{deger2012statistical,reimer2012modeling,nawrot2008measurement,shinomoto2003differences,maimon2009beyond,shinomoto2009relating}
for examples in Neuroscience).
However, the stochastic integration methods that provide the basis of our results allow the treatment of random signals and intensities, provided they are predictable, which encompasses a wide enough class of processes to cover most applications \citep{protter2005stochastic}.  
Most of our results thus have straightforward generalizations (1) to the case of random continuous signal, with the differences that the variance of the estimates would increase due to the additional variability induced by the signal fluctuations, (2) to the case of random intensities, but the expressions obtained would depend on the statistical properties of $\lambda(t)$, which may or may not have simple analytical expressions.
As a potential direction for extension of the framework,
Hawkes process \citep{hawkes1971point} is a point process
wherein the probability of occurrence of future events can also depend on the sequence of events happened in the past.
Indeed, due this history dependency it is also called a self-exciting process.
Hawkes process is being used for modeling recurrent interactions in various fields,
for instance in finance it is used to model buy or sell transaction events on stock market \citep{embrechts2011} 
in geology to model the origin times and magnitudes of earthquakes \citep{ogata1988},
in online social media to model user actions over time \citep{rizoiu2017}
and even modeling reliability of information on the web and controlling the spread misinformation
\citep{tabibian2017distilling,kim2018leveraging}, and
in Neuroscience to model spike trains \citep{kruminCorrelationbasedAnalysisGeneration2010}.
If we recast our counting process $N(t)$ (Eq.(\ref{eq:cpm}))
to incorporate the history dependency,
in principle we should be able to extend our theoretical results beyond Poisson statistics.

\subsubsection*{Extension beyond binary significance assessment}

We show that the Marchenko-Pastur distribution provides a good approximation of the distribution of eigenvalues in absence of coupling, and the  upper end of its support approximates the largest eigenvalue. This provides us a threshold to assess the significance of empirical singular values. Nevertheless, this hard thresholding approach  does not take into account the actual fluctuations of the largest eigenvalue around this upper end of the support, and thus does not provide of meaningful p-value for the statistical test.

It has been shown that the appropriately rescaled and recentered
\footnote{
  Required recentering and rescaling of the eigenvalues is elaborated
  in \cite{johnstoneDistributionLargestEigenvalue2001,karouiLargestEigenvalueWishart2003,karouiTracyWidomLimit2007}
}
largest eigenvalue of Wishart matrices is asymptotically distributed as the Tracy-Widom distribution
(for example see
\cite{johnstoneDistributionLargestEigenvalue2001,tracyDistributionFunctionsLargest2002,karouiLargestEigenvalueWishart2003,karouiRecentResultsLargest2005,karouiTracyWidomLimit2007},
however note that in some cases of practical relevance, the normal distribution might be more appropriate \citep{baiCentralLimitTheorems2008}).
Such asymptotic distribution of the largest eigenvalue can be exploited for reporting a theoretical p-value for the significance of the coupling
and therefore extending the significance assessment from a binary decision to a probabilistic one.
For example, \cite{kritchmanNonParametricDetectionNumber2009} exploit this idea 
(but in a simpler scenario) to determine the number of the signal component in noisy data.
This extension would allow a precise probabilistic assessment of the significance of weaker couplings leading to eigenvalues in the neighborhood of the asymptotic threshold introduced above.
%



%% file: neuralComp_conclusion.tex
We investigated the statistical properties of coupling
measures between continuous signals and point processes.
We first used martingale theory to characterize the distributions univariate coupling measures such as the
 $\rm PLV$, 
then, based on multivariate extensions of this result and Random Matrix Theory,
we establish predictions regarding the null distribution of the singular values of coupling matrices between a large number of point processes and continuous signals, and a principled way to assess significance of such multivariate coupling.
These theoretical results build a solid basis for the statistical assessment of such coupling in applications dealing with high dimensional data.



%% file: neuralComp_appendix.tex

\section{Proofs of main text theorems}\label{app:proof}
\begin{proof}[Proof of Theorem~\ref{thm:1Dcoupling}]
For the first part of the Theorem (expectation), 
we use the martingale $M^{(k)}$ associated to each copied processes $N^{(k)}$ to rewrite 
\begin{equation}\label{eq:sumMartingale}
\widehat{\text{\rm c}}_K =\frac{1}{K}\sum_{k = 1}^K\int_0^T x(t)dM^{(k)}(t)+\frac{1}{K}\sum_{k = 1}^K\int_0^T x(t)\lambda (t)dt(t)\,.
\end{equation}
As explained above, elements of the sum in the first term is then zero mean martingale, 
and by linearity so is the whole term. 
As a consequence (using the zero mean property), 
the expectation of the first term is zero and only remains the second term
\[
  \mathbb{E}\left[\widehat{\text{\rm c}}_K\right]=\int_0^T x(t)\lambda (t)dt(t)\,.
\]

We then exploit a Central Limit Theorem (CLT) for martingales to prove the second part of the theorem (convergence to Gaussian distribution). 
To satisfy the CLT in such case, 
it is sufficient to find a particular martingale $\widetilde{M}^{(K)}$ sequence that will  satisfy the conditions described in
\citep[p. 63]{aalen2008SurvivalEventHistory} 
($\overset{P}{\rightarrow}$ indicate convergence in probability):
\begin{enumerate}
\item[(1)]
  $\text{Var}(\widetilde{M}^{(K)}(t))\overset{P}{\underset{K\rightarrow+\infty}{\longrightarrow}} \widetilde{V}(t)$ for all $t\in[0,\,T]$, 
  with $\widetilde{V}$ increasing and $\widetilde{V}(0)=0$,
\item[(2)]  
  informally, the size of the jumps of $\widetilde{M}^{(K)}$ tends to zero 
  (see \citet[p. 63]{aalen2008SurvivalEventHistory}). 
  Formally, for any $\epsilon>0$, 
  the martingale $\widetilde{M}_\epsilon^{(K)}(t)$ gathering the jumps $>\epsilon$ satisfies $\mbox{Var}
  \left(\widetilde{M}_\epsilon^{(K)}(t)\right)\overset{P}{\underset{K\rightarrow+\infty}{\longrightarrow}}0$.
\end{enumerate}
Then $\widetilde{M}^{(K)}(t)$ converges in distribution to a Gaussian martingale of variance $\widetilde{V}(t)$.

To achieve these conditions, let us define ${M}^{(k)}$,
the sequence of independent identically distributed zero mean martingales defined on $[0,\,T]$
canonically associated to the point process of each trial $N^{(k)}$.
Then we build martingales $M_x^{(k)}(t)=\int_0^t
x(s)dM^{(k)}(s)ds$ and construct $\widetilde{M}^{(K)}=1/\sqrt{K} \sum_{k=1}^K M_x^{(k)}$.
 
The variance of this later martingale (also called its predictable variation process) 
can be computed based on the rules provided in Appendix~\ref{app:Hmartingale}. 
First due to trial independence
\begin{align}
\label{eq:varMartingaleTot}
\widetilde{V}(t)  = \text{Var}\left(\widetilde{M}^{(K)}(t)\right) = \text{Var}\left(\frac{1}{\sqrt{K}}\sum_{k=1}^{K} M_x^{(k)}(t)\right)
 = \sum_{k=1}^{K} \text{Var}\left(\frac{1}{\sqrt{K}} M_x(t)\right)\,,
\end{align}
 and using Eq.(\ref{eq:martinVar}), we get
\begin{equation}
\label{eq:varMartingaleTot_final}
\widetilde{V}(t)=\frac{1}{K} \sum \int_0^t x^2(t)\lambda(t)dt= \int_0^t x^2(t)\lambda(t)dt\,.
\end{equation}
Eq.(\ref{eq:varMartingaleTot_final}) clearly fulfills CLT's condition (1). 

For condition (2), due to Assumption~\ref{assum:deter}, 
$x(t)$ is bounded, such that there is a $B>0$ satisfying $|x(t)|<B$ over $[0,\,T]$. 
As a consequence, the size of all jumps is bounded by $B/\sqrt{K}$, 
and for any $\epsilon$, $\widetilde{M}_\epsilon^{(K)}(t)$ is the constantly zero for    $K>\frac{B^2}{\epsilon^2}$ and condition (2) is satisfied.

Fulfillment of both conditions lead to convergence in distribution to a Gaussian martingale of variance $\widetilde{V}(t)$,
\begin{equation}
\label{eq:normalDistMartingale}
\widetilde{M}^{(K)}\underset{K\rightarrow +\infty
}{\longrightarrow}\mathcal{N}\left(0,\int_0^T x^2(t)\lambda(t)dt\right)\,.
\end{equation}

Finally, using Eq.(\ref{eq:sumMartingale}) we conclude the proof by noticing that the above martingale corresponds exactly to the
quantity $\sqrt{K}\left(\widehat{\text{\rm c}}_K-\text{\rm c}^*\right)$.
 Therefore,

  \begin{equation}
    \label{eq:normalDistMartingale}
    \sqrt{K}\left(\widehat{\text{\rm c}}_K-\text{\rm c}^*\right)\underset{K\rightarrow +\infty
    }{\longrightarrow}\mathcal{N}\left(0,\int_0^T x^2(t)\lambda(t)dt\right)\,.
  \end{equation}

\end{proof}

\begin{proof}[Proof of Corollary~\ref{corol:1DPLV}]
  We apply Theorem~\ref{thm:1Dcoupling} to $e^{\boldsymbol{i}\phi(t)}$
  (\ie replacing $x(t)$ with $e^{\boldsymbol{i}\phi(t)}$).
  As $e^{\boldsymbol{i}\phi(t)}$ is complex-valued we should have a covariance function for its predictable variation process $\widetilde{V}(t)$.
  The covariance between martingales real part  
  \[
    M_{\text{Re}}(t)=\int_0^t \text{Re}(e^{\boldsymbol{i}\phi(s)})dM(s)ds
  \]
  and imaginary part 
  \[
    M_{\text{Im}}(t)=\int_0^t \text{Im}(e^{\boldsymbol{i}\phi(s)})dM(s)ds
  \]  is given by
  \begin{equation}\label{eq:martngaleCov}
    \int_0^t \text{Re}(e^{\boldsymbol{i}\phi(s)})\text{Im}(e^{\boldsymbol{i}\phi(s)})\lambda(s)ds\,.
  \end{equation}
  The diagonal elements of the covariance function are the predictable variation process of
  $M_{\text{Re}}$ and $M_{\text{Im}}$ that can be computed based on Eq.(\ref{eq:martinVar})
  and the off-diagonal elements are the covariance between martingales real and imaginary part
  that can that can be computed based on Eq.(\ref{eq:martinCov}).
  Therefore, covariance function for its predictable variation process follows as,
  \begin{align}
    \label{eq:martingale2dCov}
    \text{\rm Cov}
    \left[
    \begin{matrix}
      \text{Re}\{Z\}\\
      \text{Im}\{Z\}
    \end{matrix}
    \right] 
    & =
      \left[
      \begin{matrix}
	\int_0^t \left(\text{Re}(e^{\boldsymbol{i}\phi(s)})\right)^2\lambda(s)ds & 
        \int_0^t \text{Re}(e^{\boldsymbol{i}\phi(s)})\text{Im}(e^{\boldsymbol{i}\phi(s)})\lambda(s)ds\\
	\int_0^t \text{Re}(e^{\boldsymbol{i}\phi(s)})\text{Im}(e^{\boldsymbol{i}\phi(s)})\lambda(s)ds\ &
        \int_0^t \left(\text{Im}(e^{\boldsymbol{i}\phi(s)})\right)^2\lambda(s)ds
      \end{matrix}\right]\\
    & =
      \int_0^t
      \left[
      \begin{matrix}
	\cos^2(\phi(s))&\sin(2\phi(s))/2\\
	\sin(2\phi(s))/2&	\sin^2(\phi(s))
      \end{matrix}\right]
                          \lambda(s)ds\,.
  \end{align}

  Similar to Theorem~\ref{thm:1Dcoupling}, as $K\rightarrow +\infty$,
  the residuals converges in distribution to a zero-mean complex Gaussian variable $Z$ 
  (\ie the joint distribution of real and imaginary parts is
  Gaussian).
  \[
    \sqrt{K}\left(\widehat{\text{\rm c}}_K-\text{\rm c}^*\right)
    \underset{K\rightarrow +\infty
    }{\longrightarrow}\mathcal{N}\left(0, \rm Cov(Z)\right)\,.
  \]
  Because Theorem~\ref{thm:1Dcoupling} guaranties the $\sqrt{K}(\widehat{c}_K- c^*)$ tends to a Gaussian with finite variance, 
 $\widehat{c}_K$ tends to the Dirac measure in $c^*$.

However, given that we use $x(t)=e^{\boldsymbol{i}\phi(t)}$, 
$\widehat{\text{\rm c}}_K$ is not exactly the multi-trial $\rm PLV$ estimate. More precisely,
\begin{multline*}
  \widehat{\text{\rm c}}_K = \frac{1}{K}\sum_{k = 1}^K\int_0^T
  e^{\boldsymbol{i}\phi(t)}dN^{(k)}(t)=\frac{1}{K}\sum_{k = 1}^K\sum_{j=1}^{N_k} e^{\boldsymbol{i}\phi(t_j^k)}
  = \frac{\left(\sum_{k=1}^{K}N_k\right)}{K}\widehat{\text{\rm PLV}}_K\,.
\end{multline*}
Thus we can write $\widehat{\text{\rm PLV}}_K=\nu_K\cdot \widehat{\text{\rm c}}_K$, with $\nu_K=\frac{K}{\left(\sum_{k=1}^{K}N_k\right)}$.
With the same techniques (using $x(t)=1$), we can show convergence in distribution of $\nu_K$ to a constant, 
\[
\frac{1}{\nu_K}=\frac{\left(\sum_{k=1}^{K}N_k\right)}{K} = 1/K \sum_k \int_0^T 1\cdot dN^{(k)}\underset{K\rightarrow+\infty}{\longrightarrow} \int_0^T \lambda(t)dt=\Lambda(T)\,.
\]
This leads to
\[
\nu_K \underset{K\rightarrow+\infty}{\longrightarrow} \frac{1}{\Lambda(T)}\,.
\] 
Following a version of Slutsky's theorem \citep[Theorem~5.10]{mittelhammer1996mathematical}, since $\nu_k$ and $\widehat{\text{\rm c}}_K$ tends to a limit
in distribution, and one of these limits is a constant, then the
product tends to the product of the limits 
such that we get 
\[
\text{\rm PLV}^* = \lim_{K \to \infty} \nu_K\cdot \widehat{\text{\rm c}}_K= \frac{c^*}{\Lambda(T)}\,,
\]
and can decompose the PLV residual as follows: 
  \[
  {\sqrt{K}}\left(\widehat{\text{\rm PLV}}_K-\text{\rm PLV}^*\right)={\sqrt{K}}\nu_K\left(\widehat{\text{\rm c}}_K-\text{\rm c}^*\right)+ {\sqrt{K}}\left(\nu_K\text{\rm c}^*-\text{\rm PLV}^*\right)\,. 
  \]
  Taking the limit of the above equation, the second term clearly vanishes (see the above limit of $\nu_K$), and the first term, using again the limit of products, leads to the final result: 
  \[
 {\sqrt{K}}\left(\widehat{\text{\rm PLV}}_K-\text{\rm PLV}^*\right)
  \underset{K\rightarrow +\infty
  }{\longrightarrow}\mathcal{N}\left(0, \frac{1}{\Lambda(T)^2}\rm Cov(Z)\right)\,.
  \]
 
\end{proof}

\begin{proof}[Proof of Corollary~\ref{corol:kappaplv}]
  We use the intensity function introduced in Eq.(\ref{eq:Vmrate}) in Corollary~\ref{corol:1DPLV}.
  The $\rm PLV$ asymptotic value ($\rm PLV^*$) can be derived from definition introduced in Eq.(\ref{thePLV}),
  \begin{align}
    {\rm PLV^*}
    & =  \frac{\int_0^T e^{\boldsymbol{i}\phi(t)}\lambda(t)dt}{\int_0^T \lambda(t)dt}\\
    & =\frac
      {r_o \int_{0}^{T} e^{\boldsymbol{i}\phi(t)} \exp(\kappa\cos(\phi(t)-\varphi_0))\phi'(t)dt}
      {r_o\int_{0}^{T} \exp(\kappa\cos(\phi(t)-\varphi_0))\phi'(t)dt}\,.
  \end{align}
  We change  the integration variable from $\phi(t)$ to $\theta$,
  \begin{align}
    \label{eq:theoPLVintermed1}
    {\rm PLV^*}
    & =\frac{\int_{\phi(0)}^{\phi(T)} e^{\textbf{i}\theta}
      \exp(\kappa\cos(\theta-\varphi_0))d\theta}{\int_{\phi(0)}^{\phi(T)}
      \exp(\kappa\cos(\theta-\varphi_0))d\theta}\,.
  \end{align}
  To simplify the integral (bring the $\varphi_0$ out of the integral), 
  we change  the integration variable again, from $\theta$ to $\psi$,
  ($\psi = \theta -\varphi_0$),
  \begin{align}
    {\rm PLV^*}
    & = \frac{\int_{\phi(0)-\varphi_0}^{\phi(T)-\varphi_0} e^{\boldsymbol{i}(\psi + \varphi_0)}
      \exp(\kappa\cos(\psi))d\psi}{\int_{\phi(0)-\varphi_0}^{\phi(T)-\varphi_0}
      \exp(\kappa\cos(\psi))d\psi}\\
    & = e^{\textbf{i}\varphi_0}
      \frac{\int_{\phi(0)-\varphi_0}^{\phi(T)-\varphi_0} e^{\boldsymbol{i}\psi}
      \exp(\kappa\cos(\psi))d\psi}{\int_{\phi(0)-\varphi_0}^{\phi(T)-\varphi_0}
      \exp(\kappa\cos(\psi))d\psi}\,.
  \end{align}
Given that that integrand is a $2\pi$-periodic functions (thus the integral is invariant to translations of the integration interval), we get
%
\[
  {\rm PLV^*}
  = e^{\textbf{i}\varphi_0}
  \frac{\int_{-\pi}^{\pi} e^{\boldsymbol{i}\psi} \exp(\kappa\cos(\psi))d\psi}{\int_{-\pi}^{\pi} \exp(\kappa\cos(\psi))d\psi}\,.
\]
Observing that the integrand of the denominator is even, while for the numerator the imaginary part is odd and the real part is even, we get
\[
  {\rm PLV^*}
  = e^{\textbf{i}\varphi_0}
  \frac{\int_{0}^{\pi} \cos(\psi) \exp(\kappa\cos(\psi))d\psi}{\int_{0}^{\pi} \exp(\kappa\cos(\psi))d\psi}\,.
\]
This proves the first part of the corollary (Eq.(\ref{eq:PLVmeanGeneral}).
By using the integral form of the modified Bessel functions $I_k$ \textbf{for $k$ integer}  
(see \eg \citet[p.~181]{watson1995treatise}):
  \begin{align}
    I_k(\kappa) 
    & =\frac{1}{\pi}\int_{0}^{\pi} \cos(k\theta)  \exp(\kappa\cos(\theta))d\theta+\frac{\sin(k\pi)}{\pi}\int_{0}^{+\infty} e^{-\kappa\cosh t -kt}dt \\\label{modBesselAppend}
    &=\frac{1}{\pi}\int_{0}^{\pi} \cos(k\theta)  \exp(\kappa\cos(\theta))d\theta\,,
  \end{align}
  we can derive the compact form:
  \begin{align}    
    {\rm PLV^*}
    \label{eq:theoPLVintermed2}
    & = e^{\textbf{i}\varphi_0}\frac{I_1(\kappa)}{I_0(\kappa)}\,.
  \end{align}
  
  The covariance matrix of the asymptotic distribution,
  can be easily derived by plugging Eq.(\ref{eq:Vmrate}) as $\lambda(t)$ in Corollary~\ref{corol:1DPLV}
  and integrating on $[0,\,T]$:
  \begin{equation}
    \left(\rm Cov(Z) \right)_{11}
    = \frac{\lambda_0}{\Lambda(T)^2} \int_0^T \cos^2(\phi(t))
    \exp\left(\kappa\cos(\phi(t)-\varphi_0)\right)\phi'(t)dt\,.
    \label{covz11}
  \end{equation}
  Based on the above developments, and noticing that the integration intervals corresponds to $2 \pi \alphaT$, with $\alphaT$ the number of oscillation periods, we have \[\Lambda(T)=\lambda_0 2 \alphaT\pi I_0(\kappa)=\lambda_0 2\frac{\phi(T)-\phi(0)}{2\pi}\pi I_0(\kappa)\,,
  \] such that
  \begin{equation}
    \left(\rm Cov(Z) \right)_{11}
    = \frac{1}{\lambda_0\left(\phi(T)\!-\!\phi(0)\right)^2I_0(\kappa)^2} \int_0^T \cos^2(\phi(t))
    \exp\left(\kappa\cos(\phi(t)-\varphi_0)\right)\phi'(t)dt\,.
    \label{covz11}
  \end{equation}
  
  To simplify the rest of the derivations, we transform the complex variable coordinates by using $e^{\textbf{i}\phi(t)}e^{-\textbf{i}\varphi_0}$ instead of $e^{\textbf{i}\phi(t)}$
  as predictable with respect to $\{\mathcal{F}_t\}$
  (\ie replacing $x(t)$ with $e^{\textbf{i}\phi(t)}e^{-\textbf{i}\varphi_0}$ in Theorem~\ref{thm:1Dcoupling}).
  With this change Eq.(\ref{covz11}) becomes,
  \begin{equation}
    \left(\rm Cov(Z) \right)_{11}
    = \frac{1}{\lambda_0\left(\phi(T)\!-\!\phi(0)\right)^2I_0(\kappa)^2} \int_0^T \cos^2(\phi(t)-\varphi_0)
    \exp\left(\kappa\cos(\phi(t)-\varphi_0)\right)\phi'(t)dt\,.  
  \end{equation}
  We change  variable of the integral from $\phi(t)-\varphi_0$ to $\theta$
  and use the following trigonometric identity,
  \begin{eqnarray}
    \label{eq:trigIdCos}
    \cos^2(\theta) = \frac{1}{2} \left( 1 + \cos(2\theta) \right)
  \end{eqnarray}
  to obtain
  \begin{multline*}
    \left(\rm Cov(Z) \right)_{11}
    = \frac{1}{2 \lambda_0\left(\phi(T)\!-\!\phi(0)\right)^2I_0(\kappa)^2}\int_{\phi(0)}^{\phi(T)}
    \left( 1 + \cos(2\theta) \right)
    \exp\left(\kappa\cos(\theta)\right)d\theta \\
    = \frac{1}{2 \lambda_0\left(\phi(T)\!-\!\phi(0)\right)^2I_0(\kappa)^2}\int_{\phi(0)}^{\phi(T)}
    \left( \exp\left(\kappa\cos(\theta)\right) +
      \cos(2\theta)\exp\left(\kappa\cos(\theta)\right) \right)d\theta \,.
  \end{multline*}
    Using again that the integration interval is $2\pi\alphaT$ with $\alphaT$ integer, and integrates $2\pi$-periodic functions (thus the integral is invariant to translations of the integration interval), we get
    \begin{multline*}
      \left(\rm Cov(Z) \right)_{11}
      = \frac{1}{2 \lambda_0\left(\phi(T)-\phi(0)\right)^2I_0(\kappa)^2}
      \left[
        \int_{0}^{2\pi\alphaT}
        \exp\left(\kappa\cos(\theta)\right) d\theta\right. \\
      \left.+
        \int_{0}^{2\pi\alphaT}
        \cos(2\theta)\exp\left(\kappa\cos(\theta)\right) d\theta
      \right]
    \end{multline*}  
    \begin{align}
      \left(\rm Cov(Z) \right)_{11}
      & = \frac{1}{2 \lambda_0\left(\phi(T)-\phi(0)\right)^2I_0(\kappa)^2}
        \left[
        2\alphaT\pi I_0(\kappa)+ 2\alphaT\pi I_2(\kappa)
        \right]\\
      & = \frac{2\pi\alphaT}{ 2\lambda_0\left(\phi(T)-\phi(0)\right)^2 I_0(\kappa)^2}
        \left[
        I_0(\kappa)+  I_2(\kappa)
        \right]\\
      &=\frac{1}{2\lambda_0 \left(\phi(T)-\phi(0)\right) I_0(\kappa)^2}
        \left[
        I_0(\kappa)+  I_2(\kappa)
        \right]\,.\label{eq:4}
    \end{align}
where $\alphaT$ is the number of oscillation periods contained in $[0,\,T]$.
 
  We can have a similar calculation for the imaginary part \ie $\left(\rm Cov(Z) \right)_{22}$ as well,
  but using the identity   $\sin^2(\theta) = \frac{1}{2} \left( 1 - \cos(2\theta) \right)$
  instead of Eq.(\ref{eq:trigIdCos}).
  The off-diagonal elements of the covariance matrix vanish due to symmetry of integrand.

  Therefore, we showed that for a given $\kappa \geq 0$, scaled residual
  \[
    Z' = e^{-\boldsymbol{i}\varphi_0} \sqrt{K}\left(\widehat{\rm PLV}_K - \rm PLV^* \right)\,,
  \]
  converges to a zero mean complex Gaussian with the following covariance:
  \begin{equation*}
    \text{\rm Cov}
    \left[
      \begin{matrix}
        \text{Re}\{Z'\}\\
        \text{Im}\{Z'\}
      \end{matrix}
    \right]=\left[
      \begin{matrix}
        \text{Re}\{Ze^{-i\varphi_0}\}\\
        \text{Im}\{Ze^{-i\varphi_0}\}
      \end{matrix}
    \right]=\frac{1}{2 \lambda_0 (\phi(T)\!-\!\phi(0)) I_0(\kappa)^2}\left[
      \begin{matrix}
        I_0(\kappa)+I_2(\kappa) & 0\\
        0&	I_0(\kappa)-I_2(\kappa)
      \end{matrix}\right]\,.
  \end{equation*}
\end{proof}

\input{neuralComp_coroProof_uniuncoupl_v1.tex}

\input{neuralComp_thm2proof_v1.tex}

\begin{proof}[Proof of Theorem~\ref{thm:cvMP}]
  Based on Proposition~\ref{prop:cvMPdep} in Appendix~\ref{app:rmtBackground}, 
  we need only to check the four following necessary conditions, using the Kronecker delta notation of Eq.(\ref{eq:delta})
\input{moments.tex}
\end{proof}

\begin{proof}[Proof of Theorem~\ref{thm:evCvEnds}]
	Let us write the result of	 \citep{chafai2018convergence} readapted to our complex case and adapt the dimension notation ($n\rightarrow p(n)$, $m_n\rightarrow n$, but we keep the notation $X_n$) (we checked in all proofs and lemmas that the result still hold when we replace symmetric matrices by hermitian ones and scalar product of real vectors by Hermitian products of complex vectors, putting an absolute value to the hermitian product when the original scalar product was squared). We consider $\{X_n\}$, a sequence of isotropic (i.e. identity covariance) zero mean random vectors, and consider the empirical covariance matrix that for $n$ independent copies of $X_n$,
	\[
	\widehat{\Sigma}_n = \frac{1}{n}\sum_{k=1}^n \boldsymbol{X}_n^{(k)} {\boldsymbol{X}_n^{(k)}}^H\,.
	\]
	We  rely on the Strong Tail Projection property (STP) that guaranties convergence of the spectral measure of the empirical covariance to the MP law, and convergence of the extreme eigenvalues to the ends of the MP support.
	\begin{defn}[Strong Tail Projection property (STP)]
		STP hold when there exist $f:\mathbb{N}\rightarrow [0,1]$, $g:\mathbb{N}\rightarrow \mathbb{R}^+$ such that $f(r)\rightarrow 0$ and $g(r)\rightarrow 0$ as $r\rightarrow \infty$, and for every $p\in\mathbb{N}$, for any orthogonal projection $P:\mathbb{C}^p\rightarrow\mathbb{C}^p$ of rank $r>0$, for any real $t>f(r).r$ we have
		\[
		\mathbb{P}\left(\left\Vert P \boldsymbol{X}_n\right\Vert^2-r\geq t\right)\leq \frac{g(r)r}{t^2}\,.
		\]
	\end{defn}
	
	By noting that $\mathbb{ E}\left\Vert P X_n\right\Vert ^2=r$, we can use Chebyshev's inequality to satisfy such property: let $\sigma^2$ be the variance of $\left\Vert P \boldsymbol{X}_n\right\Vert^2$, the inequality leads to, for any $t$
	\[
	\mathbb{P}\left(\left\Vert P \boldsymbol{X}_n\right\Vert^2-r\geq \sigma t\right)\leq 	\mathbb{P}\left(\left|\left\Vert P \boldsymbol{X}_n\right\Vert^2-r\right|\geq \sigma t\right)\leq \frac{1}{t^2}\,,
	\]
	so we get $\mathbb{P}\left(\left\Vert P \boldsymbol{X}_n\right\Vert^2-r\geq  t\right)\leq \sigma^2/t^2$ and just need to find an upper bound of $\sigma^2$ of the form $g(r)r$. To limit the complexity of the rank-dependent analysis, we will look for $g$ in the for $g(r)=C/r$ for a fixed positive constant, such that we just need to bound the above variance by a constant. Finer bounds are likely possible but left to future work.
	
	In our specific case, in line with proof of Theorem~\ref{thm:cvMP} we use
	\[
	\boldsymbol{X}_{n}=\int_{0}^{T}\frac{\boldsymbol{x}(t)}{\sqrt{K\lambda T}}dP(t)
	\]
	with $P$ the compensated Poisson process martingale of rate $K\lambda$. In an orthonormal basis adapted to the othogonal projection $P$ with rank $r$, we can rewrite
	\[
	\left\Vert P X_n\right\Vert^2=\sum_{k=1}^r\left|\left\langle \boldsymbol{w}_k,\, \boldsymbol{X}_{n}\right\rangle\right|^2\,, 
	\]
	where $\{\boldsymbol{w}_k\}$ are $r$ orthonormal vectors in $\mathbb{C}^p$
	Then we have
	\[
	\sigma^2 = \sum_{k,l\leq r} \mathbb{ E} \left[\left|\left\langle \boldsymbol{w}_k,\, \boldsymbol{X}_{n}\right\rangle\right|^2\left|\left\langle \boldsymbol{w}_l,\, \boldsymbol{X}_{n}\right\rangle\right|^2-1\right]\,.
	\]
	Using similar fourth order moment results as in Theorem~\ref{thm:cvMP} (based on Proposition~\ref{prop:fourthMom}) leads to an expansion for which all terms vanish but one per expectation, leading to
	\[
	\sigma^2 = \frac{1}{K\lambda T^2}\sum_{k,l\leq r}\int_{0}^{T} \left\langle \boldsymbol{w}_k,\, \boldsymbol{x}(t)\right\rangle 
	\left\langle \boldsymbol{x}(t),\, \boldsymbol{w}_k\right\rangle
	\left\langle \boldsymbol{w}_l,\, \boldsymbol{x}(t)\right\rangle
	\left\langle \boldsymbol{x}(t),\, \boldsymbol{w}_l\right\rangle
	dt \,.
	\]
	which can be rewritten using the Hermitian operator $\mathcal{X}$ acting on the space of $p\times p$ matrices as a positive definite bilinear form
	\[
	\mathcal{X}(U,V) =
	\int_{0}^{T} \left\langle V,\, \boldsymbol{x}\boldsymbol{x}^H(t)\right\rangle 
	\left\langle \boldsymbol{x}\boldsymbol{x}^H(t),\, U\right\rangle
	dt
	\]
	with associated eigenvalues $\xi_1\geq ...\geq \xi_{p^2}\geq 0$ such that
	\[
	\sigma^2 = \frac{1}{K\lambda T^2}\sum_{k,l\leq r} 	\mathcal{X}\left(\boldsymbol{w}_k\boldsymbol{w}_l^H,\boldsymbol{w}_k\boldsymbol{w}_l^H\right) \,.
	\]
	This sum is maximized when the $r^2$ unitary tensor matrices of the sum $\boldsymbol{w}_k\boldsymbol{w}_l$ are eigenvectors associated to the largest eigenvalues of the operator, such that we get
	\[
	\sigma^2 \leq \frac{1}{K\lambda T^2}\sum_{k=1\leq r^2} 	\xi_k
	\]
	which is it self upper bounded by the trace of the operator, leading to
	\[
	\sigma^2 \leq \frac{1}{K\lambda T^2}\sum_{k,l\leq p(n)}\int_{0}^{T} |x_k x_l|^2dt
	\]
	which is bounded according the theorem's assumptions, completing the proof.

\end{proof}

\section{Additional background and useful results}\label{app:addBackground}
\subsection{Jump processes}
Jump process exhibit discontinuities related to the occurrence of random events, 
which are distributed according to the given point process models. 
In this paper we will be concerned with jump times distributed according to (possibly inhomogeneous) Poisson processes. 

\subsubsection{Martingales related to counting processes}\label{app:Hmartingale}
As introduced in section \ref{sec:cpm} (Eq.(\ref{eq:cpm})),
under mild assumptions, we can associate a zero-mean martingale to a counting process $N(t)$:
\begin{equation}
  \label{eq:cpm_background}
  M(t) = N(t) - \int_0^t \lambda(s)ds\,.
\end{equation}
In addition, in our case (deterministic intensity), the variance of $M(t)$ is given by
\[
  V(t) = \mathbb{E}\left[M(t)^2\right]=  \int_0^t \lambda(s)ds\,. 
\]
\subsubsection{Stochastic integrals}
Now if we consider a deterministic predictable process $H$ 
(w.r.t. to the same filtration $\mathcal{F}_t$), the stochastic integration
\begin{equation}
  \label{eq:defCmartingale}
  M_H(t)=\int_0^t H(s)dM(s)ds\,.
\end{equation}
Using Eq.(\ref{eq:cpm_background}), we can write:
\begin{equation}
  \label{eq:splitCmartingale_background}
  M_H(t) =  \int_{0}^t H(s) dN(s) - \int_{0}^t H(s) \lambda(s) ds\,.
\end{equation}
which is equivalent to Eq.(\ref{eq:doobmeyer}) which introduced the separation of the deterministic component of empirical coupling measure from the (zero-mean) random fluctuations of the measure.
$M_H(t)$ is also a zero-mean martingale with respect to history $\{\mathcal{F}_t\}$. 
This trivially entails that $\mathbb{E}\left[M_H(t)\right]=0$ at all times.

\subsubsection{Second order statistics}
In addition, the second order statistics of such stochastic integrals can be explicitly derived from the original intensities. 
In particular for $M_H(t)=\int_0^t H(s)dM(s)ds$, we have the variance 
\begin{equation}\label{eq:martinVar}
V_H(t) =  \mathbb{E}\left[M_H(t)^2\right]=\int_0^t H(s)^2\lambda(s)ds\,,
\end{equation}
that corresponds to its \textit{predictable variation process} 
(see \citet[section 2.2.6]{aalen2008SurvivalEventHistory}). 
A similar result applies to covariance as well: 
let $G$ and $H$ be deterministic predictable, then 
\begin{equation}\label{eq:martinCov}
  V_{H,G}(t) =  \mathbb{E}\left[M_H(t)M_G(t)\right]=\int_0^t H(s)G(s)\lambda(s)ds\,. 
\end{equation}
Importantly, let us mention that this non-vanishing covariance reflects the fact that both stochastic integrals are computed from the same realization of $M(t)$. 
If two stochastic integrals are derived from independent point processes, 
the resulting covariance between them is zero.

\subsubsection{General jump stochastic processes}
For the proofs of our results, 
it is convenient to state some general results for jump processes combine deterministic and a jump stochastic integral, 
decomposable as
\begin{equation}\label{eq:jumpProc}
X(t) = X(0)+\int_0^t f(X(s),s)ds+\int_{0}^{t} h(X(s),s) dN(s)\,,	
\end{equation}
with $N(t)$ a Poisson process with intensity $\lambda(t)$, $f$ and $h$ square integrable. This clearly includes the martingales defined above. 

\subsubsection{Mean stochastic jump integrals}\label{app:meanStochInt}
According to \citet[Theorem 3.20]{hanson2007applied}, 
we can compute the expectation of $X(t)$ defined in Eq.(\ref{eq:jumpProc}).
\begin{equation}\label{eq:meanJumpInt}
\mathbb{E}[X(t)] = \mathbb{E}[X(0)] +  \int_0^t f(X(s),s)ds+\int_{0}^{t} \mathbb{E}\left[h(X(s),s)\right] \lambda(s)ds\,.
\end{equation}
This allows to retrieve the zero-mean property of the stochastic integral of martingales.

\subsubsection{It\^o's formula}\label{app:itoFormula}
It\^o's formula or It\^o's lemma is an identity to find the differential of a function of a stochastic process. 
It is a counterpart of the chain rule used to compute the differential of composed functions. 
We restrict ourselves to the case of a time independent scalar function of a jump process, 
while different formulas exist for other cases.

A generalized chain rule for the time derivative of such processes allows to derive an integral formula for scalar process $Y(t)= F(X(t))$ with $F$ continuously differentiable 
(see \citet[Lemma~4.22, Rule~4.23]{hanson2007applied}):
\begin{multline}\label{eq:itoJump}
  Y(t) = Y(0) +\int_0^t \frac{dF}{dx}(X(s))f(X(s),s)ds \\
  + \int_0^t \left[F\left(X(s_-)+h(X(s_-),s)\right) -F(X(s_-))\right]dN(s)\,,
\end{multline}
where $X(s_-)=\lim_{t\rightarrow s_-} X(t)$ indicates the left limit.

For a scalar function of a multivariate process ${Y}(t)= F(\boldsymbol{X}(t))$ with
\begin{equation}\label{eq:jumpProcMult}
  \boldsymbol{X}(t) = \boldsymbol{X}(0)+\int_0^t
  \boldsymbol{f}(\boldsymbol{X}(s),s)ds+\int_{0}^{t}
  \boldsymbol{h}(\boldsymbol{X}(s),s) dN(s)\,,	
\end{equation}
the generalization is straightforward
\begin{multline}\label{eq:itoJumpMult}
  Y(t) = Y(0) +\int_0^t \sum_k
  \frac{dF}{dx_k}(\boldsymbol{X}(s))f_k(\boldsymbol{X}(s),s)ds \\
  + \int_0^t \left[F\left(\boldsymbol{X}(s_-)+\boldsymbol{h}(\boldsymbol{X}(s_-),s)\right) -F(\boldsymbol{X}(s_-))\right]dN(s)\,.
\end{multline}
This allows retrieving the expression of martingale second order statistics presented above, 
as well as computing higher order moments required in the proof of Theorem~\ref{thm:cvMP}.

An application of this formula that we will use is the following
\begin{prop}\label{prop:fourthMom}
	Assume $W(t)=\int_0^t A(s)dM(s)$, $X(t)=\int_0^t B(s)dM(s)$, $Y(t)=\int_0^t C(s)dM(s)$, $Z(t)=\int_0^t D(s)dM(s)$ are stochastic s of the same (possibly inhomogeneous) Poisson process martingale $M(t)= N(t)-\int_{0}^{t}\lambda(s) ds$ with intensity $\lambda(t)$. Then
	\begin{multline}\label{eq:4mom}
	\mathbb{E}\left[WXYZ\right](t)=\int_{0}^{t} ABCD(s_-)\lambda(s)ds+\left(\int_{0}^{t} AB(s)\lambda(s)ds\right)\left(\int_{0}^{t} CD(s)\lambda(s)ds\right)\\+\left(\int_{0}^{t} AC(s)\lambda(s)ds\right)\left(\int_{0}^{t} BD(s)\lambda(s)ds\right)+\left(\int_{0}^{t} AD\lambda(s)(s)ds\right)\left(\int_{0}^{t} BC(s)\lambda(s)ds\right)\,.
	\end{multline}
\end{prop}
\begin{proof}
	We apply the above formula to $F(W,X,Y,Z)=WXYZ$, yielding
	\begin{multline*}
	WXYZ(t)=-\int_{0}^{t} \left(AXYZ(s)+WBYZ(s)+WXCZ(s)+WXYD(s)\right)\lambda ds\\
	+\int_{0}^{t} \left[(W(s_-)+A)(X(s_-)+B)(Y(s_-)+C)(Z(s_-)+D) - WXYZ(s_-)\right]dN(s)\,.
	\end{multline*}
	Expanding the second term we obtain the formula
	\begin{multline*}
	WXYZ(t)=\int_{0}^{t} \left(AXYZ(s)+WBYZ(s)+WXCZ(s)+WXYD(s)\right)dM(s)\\
	+\int_{0}^{t} \!\!\!\left( ABYZ(s_-)\!+\!AXCZ(s_-)\!+\!AXYD(s_-)\!+\!WBCZ(s_-)\!+\!WBYD(s_-)\!+\!WXCD(s_-)\right)dN(s)\\
	+\int_{0}^{t} ABCD(s_-)dN(s)+\int_{0}^{t} \left(ABCZ(s_-)+AXCD(s_-)+ABYD(s_-)+WBCD(s_-)\right)dN(s)	\,.
	\end{multline*}
	The first and last integral terms in this last formula have vanishing expectation, the first because it is a stochastic integral of zero mean martingale $M$, the last because each term inside the integral contains only one random variable, which is itself a stochastic integral of the martingale $M$ (and thus zero mean). Thus for the expectation we get
	\begin{multline}
	\mathbb{E}\left[WXYZ\right](t)=\int_{0}^{t} ABCD(s_-)d\lambda(s)+\int_{0}^{t} \!\!\!\left( AB\mathbb{E}YZ(s_-)\!+\!AC\mathbb{E}XZ(s_-)\!\right.\\\left.+\!AD\mathbb{E}XY(s_-)\!+\!BC\mathbb{E}WZ(s_-)\!+\!BD\mathbb{E}WY(s_-)\!+\!CD\mathbb{E}WX(s_-)\right)\lambda(s) ds\,.
	\end{multline}
	Based on the It\^o integral formula, one can easily derive and expression for the expectation of each product of two variables (see Eq.(\ref{eq:productStochIntMart})), leading to, after reordering the terms
	\begin{multline}
	\mathbb{E}\left[WXYZ\right](t)=\int_{0}^{t} ABCD(s_-)d\lambda(s)+\int_{0}^{t} \!\!\!\left( AB(s_-)\int_{0}^{s} CD(u_-)\lambda(u)du\!\right.\\\left.+CD(s_-)\int_{0}^{s}AB(u_-)\lambda(u)du\!+\!AC(s_-)\int_{0}^{s}BD(u_-)\lambda(u)du\!+\!BD(s_-)\int_{0}^{s}AC(u_-)\lambda(u)du\right.\\\left.+\!AD(s_-)\int_{0}^{s}BC(u_-)\lambda(u)du\!+\!BC(s_-)\int_{0}^{s}AD(u_-)\lambda(u)du\!\!\right)\lambda(s) ds\,.
	\end{multline}
	We then observe that the terms inside the integral can be paired such that integral form of the product derivative formula ($\int f \int g=\int \left(g\int f+f\int f\right)$) can be applied, leading directly to Eq.(\ref{eq:4mom}).
\end{proof}

\subsection{Notions of convergence}\label{app:convergence}
In contrast to finite dimensional vectors, 
there is a different and non-equivalent notions of convergence for functions and random variables. 
We explain the two types of convergence encountered in this paper.
For a random variable $X$, we consider its probability measure $\mu_X$ such that
\[
  \mu_X(A)=P(X\in A)\,,
\]
and its associated cumulative distribution function (CDF)
\[
  F_X(x)=\mu_{X}\left(\left(-\infty,x\right]\right)=P(X\leq x)
\]

\subsubsection{Convergence in distribution}
The classical definition is based on the CDF.
\begin{defn}[Convergence in distribution]
  We say that sequence of random variables $\{X_n\}$ converges in distribution (or in law) to $X$ whenever
  \[
    F_{X_n}(x)\underset{n\rightarrow+\infty}{\longrightarrow}F_X\,,
  \]
  at all continuity points of $F_X$. This is then denoted $X_n\overset{D}{\longrightarrow}X$.
\end{defn}
An equivalent definition can be formulated in terms of weak convergence:
\begin{prop}\label{prop:distCV}
  $X_n\overset{D}{\longrightarrow}X$ if and only if, for any bounded continuous function $f$,
  \[
    \mathbb{E}\left[f(X_n)\right] = \int f d\mu_{X_n} \rightarrow
    \mathbb{E}\left[f(X)\right] =\int f d\mu_X \,,
  \]
  that is, in classical topological terms, the measure $\mu_{X_n}$ converges weakly to $\mu_X$.
\end{prop}
The generalization to multidimensional variable encountered in Theorem~\ref{thm:multiclt} 
consists simply in replacing the cumulative distribution by its multivariate version, 
$F_{\boldsymbol{X}}(x) = P(X_1<x_1,\dots,X_n<x_n)$ in the above definition. 
As simple necessary and sufficient condition for $\boldsymbol{X}\longrightarrow \boldsymbol{Y}$ is that for all vectors $\boldsymbol{t}$, 
$\boldsymbol{t}^\top \boldsymbol{X}\longrightarrow \boldsymbol{t}^\top \boldsymbol{Y}$ 
(this is the Cram\'er-Wold theorem, see \citet{billingsley1995probability}).

\subsubsection{Convergence in probability}
This is a stronger notion of convergence, 
denotes $X_n\overset{P}{\longrightarrow} X$ stating that for any $\epsilon>0$
\begin{equation}\label{eq:cvProb}
  P\left(|X_n-X| >\epsilon\right)\underset{n\rightarrow +\infty}{\longrightarrow} 0\,.
\end{equation}

It can be shown that convergence in probability implies convergence in distribution. The converse is true only in special cases such as
\begin{prop}
  If $X$ converges in distribution to a (deterministic) constant $c$, 
  then it also converges to it in probability.
\end{prop}

An extension to the multivariate case is obtained in finite vector spaces by replacing the absolute value in Eq.(\ref{eq:cvProb}) by any norm, 
or simply by requiring the convergence of all components individually.

\subsubsection{Convergence of random measures}
The ESDs are random measures, and as such, random variables leaving in an infinite dimensional space of measures. 
This means that for a fixed realization $\omega$, 
the random measure $\boldsymbol{\mu}$ takes deterministic value $\boldsymbol{\mu}(\omega)$.

Several types of convergence can be defined. 
First, the notion of \textit{convergence  weakly in probability} can
be seen as a combination of the above definitions. 
It is known that the weak convergence of deterministic measures 
(see Proposition~\ref{prop:distCV}) 
can be associated to a (non-unique) metric 
(the topological space of weak convergence is metrizable). 
Let us pick such a metric $\rho(\mu,\nu)$ between two deterministic measures, then

\begin{defn}[Convergence weakly in probability]
  The sequence of random measures $\boldsymbol{\mu}_n$ converges \textit{weakly in probability} to the deterministic measure $\nu$ for any $\epsilon>0$
  \begin{equation}\label{eq:cvWeakProb}
    P\left(\rho (\boldsymbol{\mu}_n,\nu) >\epsilon\right)\underset{n\rightarrow +\infty}{\longrightarrow} 0\,.
  \end{equation}
\end{defn}

Next, we can also define convergence with probability 1 (also called almost sure convergence).
\begin{defn}[Convergence (weakly) with probability one.]
  The sequence of random measures $\boldsymbol{\mu}_n$ converges weakly with probability one to the deterministic measure $\nu$ for any $\epsilon>0$
  \begin{equation}\label{eq:cvWeakProb}
    P\left(\rho (\boldsymbol{\mu}_n(\omega),\nu) \underset{n\rightarrow +\infty}{\longrightarrow} 0\right)=1\,.
  \end{equation}
\end{defn}
As for the case of scalar random variables, 
convergence with probability one implies convergence in probability.

\subsection{Random matrix theory}\label{app:rmtBackground}
\subsubsection{Wishart ensemble}\label{app:wishartBackground}
\input{neuralComp_app_wishartEnsemble.tex}

\subsubsection{Stieltjes transform of ESD}
The Stieltjes transform is a very useful tools to establish the convergence of ESD and determine its limit. 
The Stieltjes transform of a measure $\mu$ is defined as 
\[
  m_{\mu}(z)=\int \frac{1}{x-z} d\mu(x)\,, z\in \mathbb{C}\setminus \mathbb{R}\,.
\]
A key example for us is the Stieltjes transform of the MP law, that writes
\[
  m(z)=\frac{1-c-z+\sqrt{(1+c-z)^2-4c}}{2cz}\,.
\]
Many important results relate measures to their Stieltjes transform. 
For our needs, we only need the property that the Stieltjes transform identifies the limit of a sequences of measures, 
with the following proposition that immediately derives from 
\citet[Theorem~2.4.4]{anderson2010introduction}.
\begin{prop}\label{prop:StieltjesUnique}
  If two sequences of random measures $\{\boldsymbol{\mu}_k\}$ and 
  $\{\boldsymbol{\nu}_k\}$ converge weakly in probability to a deterministic with identical Stieltjes transform, they converge to the same measure.
\end{prop}

\subsubsection{Convergence to MP for matrices with dependent coefficients}
Based on the above, we can now write a results that is a combination of results found in \citet{bai2008large} (mainly Theorem~1.1 and Corollary~1.1) adapted to our specific case. We consider a sequence of random matrices $\{\boldsymbol{X}_n\}$ with independent columns and study the ESD of 
\[
\boldsymbol{S}_n=\frac{1}{n}\boldsymbol{X}_n \boldsymbol{X}_n^H\,.
\]
In the following proposition, we use the Kronecker delta symbold delta (Eq.\ref{eq:delta})
 and denote by $\bar{X}$ the complex conjugate of $X$.
\begin{prop}\label{prop:cvMPdep}
  Let  As $n\rightarrow \infty$, assume the following. Let 
  
  \begin{enumerate}
  \item $\mathbb{E} \bar{X}_{jk}{X}_{lk} = \delta_{lj}$, for all $k$,
  \item $\frac{1}{n} \max_{j\neq l}\mathbb{E}\left| \bar{X}_{jk}{X}_{lk}-\delta_{lj}\right|^2\rightarrow 0$ uniformly in $k\leq n$,
  \item $\frac{1}{n^2}\sum_{\Gamma}\left(\mathbb{E}\left( \bar{X}_{jk}{X}_{lk}-\delta_{lj}\right)\left( {X}_{j'k}\bar{X}_{l'k}-\delta_{j'l'}\right)\right)^2\rightarrow 0$ uniformly in $k\leq n$, where 
    $\Gamma \!=\!\{(j,l,j',l'):1\leq j,l,j',l'\leq p\}\!\setminus\! \{(j,l,j',l'): j\!=\!j'\!\neq\! l\!=\!l' \mbox{ or }j\!=\!l'\!\neq\! j'\!=\!l\}$,
  \item $p/n\rightarrow \alpha\in (0,\infty)$.
  \end{enumerate}
  Then, with probability 1, the ESD of $\boldsymbol{S}_n$ tends (weakly) to the MP law.
\end{prop}
\begin{proof}[Sketch of the proof]
	We use Theorem~1.1 from \citet{bai2008large} combined with sufficient condition of Corollary~1.1, assuming the identity matrix $\boldsymbol{T}_n$. These conditions are compatible with the case of the Wishart ensemble, such that the ESD convergence to a distribution with the same Stieltjes transform as the MP law
\footnote{This requires checking that the self consistency equation (1.1) in \citet{bai2008large} has a unique solution, which they establish by equation (1.2))}. 
As a consequence of Proposition~\ref{prop:StieltjesUnique}, 
we get that the limit ESD should is the MP law.
\end{proof}


\input{fig2_MPlawIllust.tex}




\section{Additional corollaries}\label{app:addColos}
Additional corollaries based on simplifying Assumption~\ref{assum:linearPhase},
where a linear phase is considered instead of the gneral assumption on phae that was used in Corollary~\ref{corol:kappaplv} and Corollary~\ref{corol:uniuncoupl}.

\begin{assum}\label{assum:linearPhase}
  Assume that $\phi(t)$ is a linear function of $t$ on $[0,\,T]$,
  \begin{equation}
    \phi(t) = m t \,, \quad m = 2\pi f = 2\pi/\tau\,,
  \end{equation}
  where $f > 0$
  (interpretable as the frequency of an oscillation for the continuous signal)
  and  $\alphaT$ is the ratio of length ($T$) of signal to period of oscillation $\tau$,
\[
  \alphaT = \frac{T}{\tau} = \frac{\phi(T)\!-\!\phi(0)}{2 \pi} \,.
\]
\end{assum}

\begin{corol}\label{corol:kappaplv_linearPhase}
  Under the assumptions of Corollary~\ref{corol:kappaplv},
  assume additionally Assumption~\ref{assum:linearPhase} is also satisfied,
  and the intensity of the point-process is given by
  \begin{equation}
    \label{eq:Vmrate_xPhiDer}          
    \lambda(t)=\lambda_0\exp(\kappa\cos(\phi(t)-\varphi_0))\,,
  \end{equation}
  for a given $\kappa \geq 0$,
  then the expectation of the multi-trial {\rm PLV} estimate converges (for $K\rightarrow +\infty$) to
  \begin{equation}\label{eq:PLVmeanGeneral_lph}
    \text{\rm PLV}^*=
    \frac
      {\int_{0}^{T} e^{\boldsymbol{i}2\pi f t} \exp(\kappa\cos(2\pi f t \!-\! \varphi_0))dt}
      {\int_{0}^{T} \exp(\kappa\cos(2\pi f t \!-\! \varphi_0))dt}\,.
  \end{equation}
  If in addition $[0,\,T]$ corresponds to an integer number $\alphaT>0$ of periods of the oscillation
  \begin{equation}\label{eq:classicPLVbessel_simplified}
    \text{\rm PLV}^* =	e^{\boldsymbol{i}\varphi_0}\frac{\int_{\phi(0)}^{\phi(T)} \cos(\theta)  \exp(\kappa\cos(\theta))d\theta}{\int_{\phi(0)}^{\phi(T)}  \exp(\kappa\cos(\theta))d\theta}=e^{\boldsymbol{i}\varphi_0}\frac{I_1(\kappa)}{I_0(\kappa)}\,,
  \end{equation}    
  and the scaled residual
  $\sqrt{K}\left(\widehat{\text{\rm PLV}}_K - \text{\rm PLV}^* \right)$
  converges to a zero mean complex Gaussian $Z$ with the following covariance:
  \begin{equation}
    \label{eq:covPLSPT}
    \text{\rm Cov}
    \left[
      \begin{matrix}
        \text{Re}\{Ze^{-i\varphi_0}\}\\
        \text{Im}\{Ze^{-i\varphi_0}\}
      \end{matrix}
    \right]=\frac{1}{2\lambda_0 T I_0(\kappa)^2}\left[
      \begin{matrix}
        I_0(\kappa)\!+\!I_2(\kappa) & 0\\
        0&	I_0(\kappa)\!-\!I_2(\kappa)
      \end{matrix}\right]\,.
  \end{equation}
\end{corol}

\begin{proof}[Proof]
  We use the intensity function introduced in Eq.(\ref{eq:Vmrate_xPhiDer}). 
  The $\rm PLV$ asymptotic value ($\rm PLV^*$) can be derived from definition introduced in Eq.(\ref{thePLV}) by using Assumption~\ref{assum:linearPhase},
  \begin{align}
    {\rm PLV^*}
    & =  \frac{\int_0^T e^{\boldsymbol{i}\phi(t)}\lambda(t)dt}{\int_0^T \lambda(t)dt}\\
    & =\frac
      {\lambda_0 \int_{0}^{T} e^{\boldsymbol{i}\phi(t)} \exp(\kappa\cos(\phi(t)-\varphi_0))dt}
      {\lambda_0\int_{0}^{T} \exp(\kappa\cos(\phi(t)-\varphi_0))dt}\\
    & =\frac
      {\lambda_0 \int_{0}^{T} e^{\boldsymbol{i}mt} \exp(\kappa\cos(mt-\varphi_0))dt}
      {\lambda_0\int_{0}^{T} \exp(\kappa\cos(mt-\varphi_0))dt}\,.
  \end{align}
  We change  the integration variable from $mt$ to $\theta$,
  \begin{align}
    \label{eq:theoPLVintermed1}
    {\rm PLV^*}
    & =\frac{\int_{\theta(0)}^{\theta(T)} e^{\boldsymbol{i}\theta} \exp(\kappa\cos(\theta-\varphi_0))d\theta}{\int_{\theta(0)}^{\theta(T)} \exp(\kappa\cos(\theta-\varphi_0))d\theta}\,.
  \end{align}
  To simplify the integral (bring the $\varphi_0$ out of the integral), 
  we change  the integration variable again, from $\theta$ to $\psi$,
  ($\psi = \theta -\varphi_0$),
  \begin{align}
    {\rm PLV^*}
    & = \frac{\int_{\theta(0)-\varphi_0}^{\theta(T)-\varphi_0} e^{\boldsymbol{i}(\psi + \varphi_0)} \exp(\kappa\cos(\psi))d\psi}{\int_{\theta(0)-\varphi_0}^{\theta(T)-\varphi_0} \exp(\kappa\cos(\psi))d\psi}\\
    & = e^{\boldsymbol{i}\varphi_0}
      \frac{\int_{\theta(0)-\varphi_0}^{\theta(T)-\varphi_0} e^{\boldsymbol{i}\psi} \exp(\kappa\cos(\psi))d\psi}{\int_{\theta(0)-\varphi_0}^{\theta(T)-\varphi_0} \exp(\kappa\cos(\psi))d\psi}\,.
  \end{align}
When $[0,\,T]$ corresponds to an integer number  of periods of the oscillation
(\ie is an integer number),
and given that the integration interval is $2\pi\alphaT$, 
and integrates $2\pi$-periodic functions 
(thus the integral is invariant to translations of the integration
interval), we have
%
\[
 {\rm PLV^*}
= e^{\boldsymbol{i}\varphi_0}
\frac{\int_{-\pi}^{\pi} e^{\boldsymbol{i}\psi} \exp(\kappa\cos(\psi))d\psi}{\int_{-\pi}^{\pi} \exp(\kappa\cos(\psi))d\psi}\,.
\]
Observing that the integrand of the denominator is even, while for the numerator the imaginary part is odd and the real part is even, we get
\[
{\rm PLV^*}
= e^{\boldsymbol{i}\varphi_0}
\frac{\int_{0}^{\pi} \cos(\psi) \exp(\kappa\cos(\psi))d\psi}{\int_{0}^{\pi} \exp(\kappa\cos(\psi))d\psi}\,.
\]
We  prove the first part of the corollary (Eq.(\ref{eq:PLVmeanGeneral_lph}).
By using the integral form of the modified Bessel functions $I_k$ \textbf{for $k$ integer}  (see \eg \citet[p.~181]{watson1995treatise}):
  \begin{align}
    I_k(\kappa)&=\frac{1}{\pi}\int_{0}^{\pi} \cos(k\theta)  \exp(\kappa\cos(\theta))d\theta+\frac{\sin(k\pi)}{\pi}\int_{0}^{+\infty} e^{-\kappa\cosh t -kt}dt \\\label{modBesselAppend}
    &=\frac{1}{\pi}\int_{0}^{\pi} \cos(k\theta)  \exp(\kappa\cos(\theta))d\theta\,,
  \end{align}
  we can derive the compact form:
  \begin{align}    
    {\rm PLV^*}
    \label{eq:theoPLVintermed2}
& = e^{\boldsymbol{i}\varphi_0}\frac{I_1(\kappa)}{I_0(\kappa)}\,.
  \end{align}
  
  The covariance matrix of the asymptotic distribution,
  can be easily derived by plugging Eq.(\ref{eq:Vmrate_xPhiDer}) as $\lambda(t)$ in Corollary~\ref{corol:1DPLV}
  and integrating on $[0,\,T]$:
  \begin{equation}
    \label{covz11_linPh}
    \left(\rm Cov(Z) \right)_{11}
    = \frac{\lambda_0}{\Lambda(T)^2} \int_0^T \cos^2(\phi(t))
    \exp\left(\kappa\cos(\phi(t)-\varphi_0)\right)dt\,.
  \end{equation}
  As we have
  \[
    \Lambda(T)=\lambda_0  T I_0(\kappa) 
    \,,
  \]
  We can continue with Eq.(\ref{covz11_linPh}) as,
  \begin{equation}
  \left(\rm Cov(Z) \right)_{11}
  = \frac{1}{\lambda_0T^2I_0(\kappa)^2} \int_0^T \cos^2(\phi(t))
  \exp\left(\kappa\cos(\phi(t)-\varphi_0)\right)dt\,.
  \label{covz11_seconstep_linPh}
  \end{equation}
  
  To simplify the rest of the derivations, we transform the complex variable coordinates by using $e^{\boldsymbol{i}\phi(t)}e^{-\boldsymbol{i}\varphi_0}$ instead of $e^{\boldsymbol{i}\phi(t)}$
  as predictable with respect to $\{\mathcal{F}_t\}$
  (\ie replacing $x(t)$ with $e^{\boldsymbol{i}\phi(t)}e^{-\boldsymbol{i}\varphi_0}$ in Theorem~\ref{thm:1Dcoupling}).
  With this change Eq.(\ref{covz11_seconstep_linPh}) becomes,
  \begin{equation}
    \left(\rm Cov(Z) \right)_{11}
    = \frac{1}{\lambda_0T^2I_0(\kappa)^2} \int_0^T \cos^2(\phi(t)-\varphi_0)
    \exp\left(\kappa\cos(\phi(t)-\varphi_0)\right)dt\,.  
  \end{equation}
  Then we change the variable of the integral from $mt-\varphi_0$ to $\theta$
  (and consequently $dt$ to $\frac{1}{m} d\theta$)
  and use the following trigonometric identity,
  \begin{eqnarray}
    \cos^2(\theta) = \frac{1}{2} \left( 1 + \cos(2\theta) \right)
  \end{eqnarray}
  to obtain
  \begin{multline*}
    \left(\rm Cov(Z) \right)_{11}
    = \frac{1}{2 m \lambda_0T^2I_0(\kappa)^2}\int_{\theta(0)}^{\theta(T)}
    \left( 1 + \cos(2\theta) \right)
    \exp\left(\kappa\cos(\theta)\right)d\theta \\
    = \frac{1}{2 m \lambda_0T^2I_0(\kappa)^2}\int_{\theta(0)}^{\theta(T)}
    \left( \exp\left(\kappa\cos(\theta)\right) +
      \cos(2\theta)\exp\left(\kappa\cos(\theta)\right) \right)d\theta \,.
  \end{multline*}
  Given that the integral is invariant to translations of the integration, we get
  \begin{multline*}
    \left(\rm Cov(Z) \right)_{11}
    = \frac{1}{2 m \lambda_0T^2I_0(\kappa)^2}
    \left[
      \int_{0}^{2\pi \alphaT}
      \exp\left(\kappa\cos(\theta)\right) d\theta\right. \\
    \left.+
      \int_{0}^{2\pi\alphaT}
      \cos(2\theta)\exp\left(\kappa\cos(\theta)\right) d\theta
    \right]
  \end{multline*}  
  \begin{align}
    \left(\rm Cov(Z) \right)_{11}
    & = \frac{1}{2 m \lambda_0T^2I_0(\kappa)^2}
      \left[
      2\alphaT\pi I_0(\kappa)+ 2\alphaT\pi I_2(\kappa)
      \right]\\
    & = \frac{2 \pi\alphaT}{ 2 m \lambda_0T^2 I_0(\kappa)^2}
      \left[
      I_0(\kappa)+  I_2(\kappa)
      \right]\\
    & = \frac{mT}{ 2 m \lambda_0T^2 I_0(\kappa)^2}
      \left[
       I_0(\kappa)+  I_2(\kappa)
      \right]\\
    &=\frac{1}{2 \lambda_0 T I_0(\kappa)^2}
      \left[
      I_0(\kappa)+  I_2(\kappa)
      \right]\,.\label{eq:4}
  \end{align}

  We can have a similar calculation for the imaginary part \ie $\left(\rm Cov(Z) \right)_{22}$ as well,
  but using the identity   $\sin^2(\theta) = \frac{1}{2} \left( 1 - \cos(2\theta) \right)$
  instead of Eq.(\ref{eq:trigIdCos}).
  The off-diagonal elements of the covariance matrix vanish due to symmetry of integrand.

  Therefore, we showed that for a given $\kappa \geq 0$, scaled residual
  \[
   Z' = e^{-\boldsymbol{i}\varphi_0} \sqrt{K}\left(\widehat{\rm PLV}_K - \rm PLV^* \right)\,,
  \]
  converges to a zero mean complex Gaussian with the following covariance:
  \begin{equation*}
    \text{\rm Cov}
    \left[
    \begin{matrix}
    \text{Re}\{Z'\}\\
    \text{Im}\{Z'\}
    \end{matrix}
    \right]=\left[
      \begin{matrix}
        \text{Re}\{Ze^{-i\varphi_0}\}\\
        \text{Im}\{Ze^{-i\varphi_0}\}
      \end{matrix}
    \right]=\frac{1}{2 \lambda_0 T I_0(\kappa)^2}\left[
      \begin{matrix}
        I_0(\kappa)+I_2(\kappa) & 0\\
        0&	I_0(\kappa)-I_2(\kappa)
      \end{matrix}\right]\,.
  \end{equation*}
\end{proof}

\begin{corol}
  \label{corol:uniuncoupl_linearPhase}
  Assume $\phi(t)=2\pi kt/T$, with $k > 0$ integer, and a sinusoidal modulation of the intensity at frequency $m/T$, with $m > 0$ integer possibly different from $k$, phase shift $\varphi_0$ and modulation amplitude $\varkappa$ such that
  \begin{equation}
    \label{eq:simpleRateCoro5}
    \lambda(t)=\lambda_0\left(1+ \varkappa \cos\left( 2\pi m t / T-\varphi_0 \right)\right),\, \lambda_0>0,\,0\leq\varkappa\leq 1\,,
  \end{equation}
  and the point process is homogeneous Poisson with rate $\lambda_0$.
  Then the expectation of the PLV estimate converges (for $K\mapsto +\infty$) to
  \begin{equation}
    \label{eq:theoPLV_simpleRateCoro5}
    \text{\rm PLV}^*=\frac{1}{2}\varkappa e^{\boldsymbol{i}\varphi_0}\delta_{km}\,,
  \end{equation}
where $\delta_{km}$ denotes the Kronecker symbol. Moreover the asymptotic covariance of $Z=\sqrt{K}\left(\widehat{\rm PLV}_K - \rm PLV^* \right)$ is
  \begin{equation}
    \text{\rm Cov}
  \left[
  \begin{matrix}
  \text{Re}\{Z\}\\
  \text{Im}\{Z\}
  \end{matrix}
  \right]=\frac{1}{2 \lambda_0 T}\left[
  \begin{matrix}
  1&0\\
  0&	1
  \end{matrix}\right]\,.
  \end{equation}
\end{corol}
\input{neuralComp_modifiedCoro5proof.tex}

\section{Circular noise}\label{app:circAddNoise}
We use random numbers drawn from the von Mises distribution to generate noise for the phase of an oscillation.
Consider the oscillation $O^{orig}[t] = e^{2 \pi \boldsymbol{i} f t}$, where the bracket indicates the oscillation is sampled at equispaced discrete times $t=\{k\Delta\}_{k=1,\dots,q}$, 
then $O[t]$ is a noisy version of this oscillation which is 
perturbed in the phase:
\begin{equation}
  \label{eq:noisyosc}
  O[t] = e^{2 \pi \boldsymbol{i} f t} \exp\left(\boldsymbol{i} \eta[t]\right)\,, 
\end{equation}
where $\eta[t]$ is sampled i.i.d. from the zero-mean von Mises distribution $\mathcal{M}(0, \kappa)$ at each time $t$. 
Notably, $\kappa$ is the dispersion parameter, therefore
larger $\kappa$ correspond to smaller variance of the noise.
In simulation used in section~\ref{ssec:simulation_multD} we use $\kappa = 10$.

In the simulation for the multivariate case,
we use $N_{osc}$-dimensional vector of oscillations,
$\boldsymbol{O}^{orig}[t] = \{O^{orig}_j[t]\}_{j = 1, \dots, N_{osc}}$, and sample i.i.d. the noise for each oscillation, leading to the vector time series $\boldsymbol{\eta}[t]$.
In this case the noisy oscillations can be written as,
\begin{equation}
  \boldsymbol{O}[t] = \boldsymbol{O}^{orig}[t] \odot \exp \left( \boldsymbol{i} \boldsymbol{\eta}[t]\right)\,,
\end{equation}
where $\odot$ is (entrywise) \emph{Hadamard} product.

The advantage of such phase noise is to preserve the spectral content of the original oscillation better
than conventional normal noise.
Nevertheless, using conventional normal (white) noise (on both real and imaginary part of the oscillation) 
did not change the results significantly.









\section{Tables of parameters}\label{tb:parmasList}

Choice of parameters used in the figures in the main text.
In all simulations, $\phi_0 = 0$.

\begin{table}[ht]
  \caption{
    Parameters used for simulations used in figure~\ref{fig:uniVarSim}
    \label{table:figParam_uniVar} 
  } 
\centering 
\begin{tabular}{c c | c  c} 
\hline\hline 
Parameter & Description & A & B \\ [0.5ex] 
\hline 
$f$ & Frequency & \multicolumn{2}{c}{1 Hz}    \\ 
$K$ & Num. of trials & \multicolumn{2}{c}{5000}  \\
$T$ & Simulation length & \multicolumn{2}{c}{5 s} \\
$\lambda_0$ & Average firing rate & \multicolumn{2}{c}{20 Hz} \\
$N_S$ & Num. of simulations & \multicolumn{2}{c}{5000}  \\
$\kappa$ & Modulation strength & 0 & 0.5 \\ [0ex] 
\hline 
\end{tabular}
\end{table}  

\begin{table}[ht]
  \caption{
    Parameters used for simulations used in figure~\ref{fig:coro3exp}
    \label{table:figParam_coro3exp} 
  } 
\centering 
\begin{tabular}{c c | c c c c } 
\hline\hline 
Parameter & Description & A & B & C & D\\ [0.5ex] 
\hline 
$f$ & Frequency & \multicolumn{4}{c}{1 Hz}    \\ 
$K$ & Num. of trials & \multicolumn{4}{c}{10}  \\
$T$ & Simulation length & 0.75 s & 0.5 s & 1 s & x-axis \\
$\lambda_0$ & Average firing rate & \multicolumn{4}{c}{30 Hz} \\
$N_S$ & Num. of simulations & \multicolumn{4}{c}{500}  \\
$\kappa$ & Modulation strength & \multicolumn{4}{c}{0} \\ [0ex] 
\hline 
\end{tabular}
\end{table}

\begin{table}[ht]
  \caption{
    Parameters used for simulations used in figure~\ref{fig:multVarSim}
    \label{table:figParam_multVarSim} 
  } 
\centering 
\begin{tabular}{c c || c c c | c c c } 
\hline\hline 
Parameter & Description & A1 & A2 & A3 & B1 & B2 & B3\\ [0.5ex] 
\hline 
$f$ & Frequency & \multicolumn{6}{c}{5 oscillatory components, 11-15 Hz}    \\ 
$K$ & Num. of trials & \multicolumn{6}{c}{10}  \\
$T$ & Simulation length & \multicolumn{6}{c}{11 s} \\
$\lambda_0$ & Average firing rate & \multicolumn{6}{c}{20 Hz} \\
$N_S$ & Num. of simulations & \multicolumn{6}{c}{100}  \\
$\kappa$ & Modulation strength & \multicolumn{3}{c}{0} & \multicolumn{3}{|c}{0.15}\\ 
$n_c$ & Num. of LFP channels & \multicolumn{6}{c}{100}  \\ [0ex] 
$n_s$ & Num. of spiking units & 10 & 50 & 90 & 10 & 50 & 90  \\ [0ex] 
$\kappa_{noise}$ & Dispersion parameter of phase noise & \multicolumn{6}{c}{10} \\ [0ex] 
\hline 
\end{tabular}
\end{table}


%% file: neuralComp_coroProof_uniuncoupl_v1.tex
\begin{proof}[Proof of Corollary~\ref{corol:uniuncoupl}]
  Similar to Corollary~\ref{corol:kappaplv}, 
  we can derive the asymptotic $\rm PLV$ (Eq.(\ref{eq:theoPLVhomoSPT})) for this case,
  from the definition in Eq.(\ref{thePLV}).
  We apply the intensity function $\lambda = \lambda_0$ in Corollary~\ref{corol:1DPLV}.
  The $\rm PLV$ asymptotic value ($\rm PLV^*$) can be derived simply by changing the integration variable from $\phi(t)$ to $\theta$
(and let $\theta \mapsto \tau(\theta)$ be its inverse).

  The covariance matrix of the asymptotic distribution,
  can be derived by the procedure we used for the proof of
  Corollary~\ref{corol:kappaplv}.
  We plug the rate $\lambda_0$ as $\lambda(t)$ in Corollary~\ref{corol:1DPLV}
  and integrate on $[0,\,T]$:
\begin{equation}
  \left(\rm Cov(Z) \right)_{11}
  = \frac{\lambda_0}{\Lambda(T)^2} \int_{0}^{T} \cos^2(\phi(t)) dt\,.
\end{equation}
By chaining the variable from $\phi(t)$ to $\theta$, we get:
\begin{equation}
  \left(\rm Cov(Z) \right)_{11}
  = \frac{\lambda_0}{\Lambda(T)^2} \int_{\phi(0)}^{\phi(T)} \cos^2(\theta) \tau'(\theta) d\theta\,.
\end{equation}
  %
  As $\Lambda(T) = \int_0^T \lambda_0 dt = \lambda_0 T$,
  we have,
  \begin{align}
    \left(\rm Cov(Z) \right)_{11}
    & = \frac{\lambda_0}{\Lambda(T)^2} \int_{\phi(0)}^{\phi(T)} \cos^2(\theta) \tau'(\theta) d\theta\\
    & = \frac{1}{\lambda_0 T^2} \int_{\phi(0)}^{\phi(T)} \cos^2(\theta) \tau'(\theta) d\theta\,.
  \end{align}
  With a similar calculation for other coefficient of the covariance matrix, we get:
  \[
    \rm Cov(Z) = 
    \frac{1}{\lambda_0 T^2} \int_{\phi(0)}^{\phi(T)}
    \left[
      \begin{matrix}
	\cos^2(\theta)&\sin(2\theta)/2\\
	\sin(2\theta)/2&	\sin^2(\theta)
      \end{matrix}\right]
    \tau'(\theta) d\theta\,.
  \]

  Therefore, we showed that the scaled residual,
  \[
    Z =  \sqrt{K}\left(\widehat{\rm PLV}_K - \rm PLV^* \right)\,,
  \]
  converges to a zero mean complex Gaussian:
  \begin{equation*}
    \sqrt{K}\left(\widehat{\rm PLV}_K - \rm PLV^* \right)
    \underset{K\rightarrow +\infty
    }{\longrightarrow}\mathcal{N}\left(
      \left[\begin{matrix}
        0\\
        0
      \end{matrix}\right], \rm Cov(Z)\right)\,.
\end{equation*}
\end{proof}


%% file: neuralComp_thm2proof_v1.tex
\begin{proof}[Proof of Theorem~\ref{thm:multiclt}]
  Similar to proof of Theorem~\ref{thm:1Dcoupling} we rely on a CLT, 
  but this time adapted to the case of vector-valued martingales 
  \citep[Appendix B]{aalen2008SurvivalEventHistory} to prove  this theorem.

  We start from the single trial empirical vector-valued coupling measure
  of Eq.~(\ref{eq:defEmpMultVarC}):
  \begin{equation}
    \label{eq:defMultVarC}
    \boldsymbol{C} = \int_0^t \boldsymbol{x}(t)d\boldsymbol{N}(t)^\top
  \end{equation}
  As for the univariate case,
  under mild assumptions, we can associate a martingale to a vector-valued counting process $\boldsymbol{N}(t)$:
  \begin{equation}
    \label{eq:cpvvm}
    \boldsymbol{M}(t) = \boldsymbol{N}(t) - \int_0^t \boldsymbol{\lambda}(s)ds\,.
  \end{equation}
  As in this theorem we assume $\boldsymbol{\lambda}(t)=\boldsymbol{\lambda}_0, t\in [0,\,T]$,
  we get,
  \begin{equation}
    \label{eq:cpvvms}
    \boldsymbol{M}(t) = \boldsymbol{N}(t) -  \boldsymbol{\lambda}_0 t\,.
  \end{equation}
  The $(p\times n)$ matrix-valued martingale for the empirical coupling matrix of Eq.(\ref{eq:defEmpMultVarC}), 
  resulting from stochastic integration, is
  \begin{equation}
    \label{eq:defCmartingale}
    \boldsymbol{M}_{\boldsymbol{x}}(t) = \int_0^t \boldsymbol{x}(s)d\boldsymbol{M}^\top(s)ds\,,
  \end{equation}
  and can be decomposed similarly to Eq.(\ref{eq:splitCmartingale_background}) as
  \begin{equation}
    \label{eq:splitVVCmartingale}
    \boldsymbol{M}_{\boldsymbol{x}}(t) =  
    \int_{0}^t \boldsymbol{x}(s) d\boldsymbol{N}(s)^\top
    - \int_{0}^t\boldsymbol{x}(s) \boldsymbol{\lambda}_0ds\,.
  \end{equation}
  By generalizing the steps of Theorem~\ref{thm:1Dcoupling}, 
  we introduced the $(p\times n)$-variate martingale
  \begin{align}
    \widetilde{\boldsymbol{M}}^{(K)}(t)
    & = 1/\sqrt{K} \sum_{k=1}^K \boldsymbol{M}_{\boldsymbol{x}}^{(k)}(t) \\
    & = 1/\sqrt{K} \sum_{k=1}^K \int_0^t \boldsymbol{x}(s)\left(d\boldsymbol{M}^{(k)}\right)^\top(s)ds\,.
      \label{eq:poolingMultVar}
  \end{align}
  We now state the CLT theorem for multivariate stochastic integral.
  \begin{prop}[Multivariate Martingale CLT \citep{aalen2008SurvivalEventHistory}, Appendix B.3]
    \label{prop:multVarMartingael}
    Given the (real) matrix valued predictable functions $\boldsymbol{H}^{(K)}(t)$, consider the multivariate stochastic integral of multivariate martingale  $\boldsymbol{M}^{(K)}$ with intensity vector $\boldsymbol{\lambda}^{(K)}(t)$
    \[
      \int_{0}^{t} \boldsymbol{H}^{(K)}(u)d\boldsymbol{M}^{(K)}(u)\,,
    \]
    Assume:
    \begin{itemize}
    \item[(1)] $\int_{0}^{t}
      \boldsymbol{H}^{(K)}(u)\mbox{diag}\{\boldsymbol{\lambda}^{(K)}(u)\}\boldsymbol{H}^{(K)}(u)^\top du
      \overset{P}{\longrightarrow}\boldsymbol{V}(t)$,
    \item[(2)] 
      $\sum_{j=1}^{k}\int_{0}^{t}(\boldsymbol{H}^{(K)}(u))^2
      \boldsymbol{1}_{|\boldsymbol{H}^{(K)}(u)|>\epsilon}\lambda_j^{(K)}(u)du
      \overset{P}{\longrightarrow}0\,,$ for all $t\in [0.T]$ and $\epsilon>0$\,.
    \end{itemize}
    Then above stochastic integral converges in distribution to a mean-zero Gaussian martingale of covariance $\boldsymbol{V}(t)$.
  \end{prop}
  %
  We notice that in summing of $K$ trials (Eq.(\ref{eq:poolingMultVar})),
  deterministic signals $\boldsymbol{x}$ remain identical and
  point processes are pooled across $K$-trials.
  Given trials are independent, 
  the counting processes derived from the pooled $K$ trials of Poisson processes $\sum_{k=1}^K \boldsymbol{N}^{(k)}(t)$ are distributed as multivariate Poisson processes with intensity vector $K\boldsymbol{\lambda}_0$, 
  such that
  \begin{equation}
    \label{eq:defCmartingalePooled}
    \widetilde{\boldsymbol{M}}^{(K)}(t)= 1/\sqrt{K} \int_0^t \boldsymbol{x}(s)d\boldsymbol{P}^\top(s)ds\,,
  \end{equation}
  where $\boldsymbol{P}$ is the martingale associated to the pooled process,
  \begin{equation}
    \label{eq:multvarPoissonMartingale}
    \boldsymbol{P}(t) = \left(\sum_{k=1}^K \boldsymbol{N}^{(k)}(t)\right) - \int_0^t K \boldsymbol{\lambda}(s)ds\,.
  \end{equation}

  Given the coupling matrix is matrix-valued, 
  we have to vectorize it in order to apply the above CLT. 
  Let $\mbox{Vec}\{.\}$ be the operator that concatenates the successive columns of a matrix into a larger column vector.
  $\widetilde{\boldsymbol{M}}^{(K)}(t)$ is $(p\times n)$-variate matrix-valued process,
  and its vectorized version, $\mbox{Vec}\{\widetilde{\boldsymbol{M}}^{(K)}(t)\}$,
a  $(pn\times 1)$-variate vector process.
  We notice that we can write Eq.(\ref{eq:defCmartingalePooled}) 
  in vectorized form as
  \[
    \mbox{Vec}\{\widetilde{\boldsymbol{M}}^{(K)}(t)\}= \int_0^t \boldsymbol{H}(s)d\boldsymbol{P}^\top(s)ds\,,
  \]
  with the $(pn\times n)$-variate block diagonal matrix
  \begin{equation}\label{eq:blockDiagMartingale}
    \boldsymbol{H}(s)=
    \frac{1}{\sqrt{K}} \left[
      \begin{matrix}
        \boldsymbol{x}(s) &\boldsymbol{0}&\cdots&\cdots&\boldsymbol{0}\\
        \boldsymbol{0} & \boldsymbol{x}(s)&\boldsymbol{0}&\cdots&\boldsymbol{0}\\
        \boldsymbol{0}&\boldsymbol{0}&\ddots& \ddots& \boldsymbol{0}\\
        \boldsymbol{0}&\ddots&\ddots& \ddots& \boldsymbol{0}\\
        \boldsymbol{0}&  \cdots&\cdots&\boldsymbol{0}&\boldsymbol{x}(s)
      \end{matrix}
    \right]\,.
  \end{equation}

  
  The variance of $\mbox{Vec}\{\widetilde{\boldsymbol{M}}^{(K)}(t)\}$ 
  (a $(pn\times pn)$-variate covariance matrix which is also called predictable variation process) 
  can be write, based on 
  Proposition~\ref{prop:multVarMartingael}, as
  \begin{equation}
    \label{eq:covVVmartSeq}
    \widetilde{\boldsymbol{V}}(t) =
    \int_{0}^{t} \boldsymbol{H}(s) \operatorname{diag}\left\{\boldsymbol{\lambda}(s) \right\} \boldsymbol{H}(s)^\top ds\,.
  \end{equation}
  Since we assume a constant intensity function,  $\boldsymbol{\lambda}(t) = \boldsymbol{\lambda}_0=\{\lambda_k\}_{k}$ ($(n\times 1)$-variate matrix),
  we can simplify Eq.(\ref{eq:covVVmartSeq}) as follows:
  \begin{align}
    \widetilde{\boldsymbol{V}}(t) & =
                                    \int_{0}^{t} \boldsymbol{H}(s) \operatorname{diag}\left\{K \boldsymbol{\lambda}_0 \right\} \boldsymbol{H}(s)^\top ds\,.
  \end{align}
  Replacing $\boldsymbol{H}(s)$ with the block diagonal matrix defined in Eq.(\ref{eq:blockDiagMartingale})
  lead us to,
  \begin{align}
    \widetilde{\boldsymbol{V}}(t) &= \frac{1}{{K}} \left[
                                    \begin{matrix}
                                      \int_{0}^{t}K\lambda_1\boldsymbol{x}(s)\boldsymbol{x}(s)^H ds &\boldsymbol{0}&\cdots&\cdots&\boldsymbol{0}\\
                                      \boldsymbol{0} & \int_{0}^{t}K\lambda_2\boldsymbol{x}(s)\boldsymbol{x}(s)^H ds&\boldsymbol{0}&\cdots&\boldsymbol{0}\\
                                      \boldsymbol{0}&\boldsymbol{0}&\ddots& \ddots& \boldsymbol{0}\\
                                      \boldsymbol{0}&\ddots&\ddots& \ddots& \boldsymbol{0}\\
                                      \boldsymbol{0}&  \cdots&\cdots&\boldsymbol{0}&\int_{0}^{t}K\lambda_n\boldsymbol{x}(s)\boldsymbol{x}(s)^H ds
                                    \end{matrix}\right]\\
                                  &= \left[
                                    \begin{matrix}
                                      \lambda_1\int_{0}^{t}\boldsymbol{x}(s)\boldsymbol{x}(s)^H ds &\boldsymbol{0}&\cdots&\cdots&\boldsymbol{0}\\
                                      \boldsymbol{0} & \lambda_2\int_{0}^{t}\boldsymbol{x}(s)\boldsymbol{x}(s)^H ds&\boldsymbol{0}&\cdots&\boldsymbol{0}\\
                                      \boldsymbol{0}&\boldsymbol{0}&\ddots& \ddots& \boldsymbol{0}\\
                                      \boldsymbol{0}&\ddots&\ddots& \ddots& \boldsymbol{0}\\
                                      \boldsymbol{0}&  \cdots&\cdots&\boldsymbol{0}&\lambda_n\int_{0}^{t}\boldsymbol{x}(s)\boldsymbol{x}(s)^H ds
                                    \end{matrix}\right]\,.
  \end{align}
This fulfills condition (1) of the CLT, 
for all $t\in[0,\,T]$. For the second condition, 
it is enough to see that coeffcients of $\boldsymbol{H}$ are bounded by a term decreasing in $\frac{1}{\sqrt{K}}$. The CLT is thus satisfied, 
and we get convergence in distribution to a zero-mean complex Gaussian of covariance $\widetilde{\boldsymbol{V}}(t)$ for each $t$. 
Specializing the result for $t=T$, we get, based on Assumption~\ref{assum:multi}, 
a diagonal covariance matrix with block-constant diagonal coefficients
\begin{equation}\label{eq:covMult}
  \widetilde{\boldsymbol{V}}(T) =
  \left[
    \begin{matrix}
      T\lambda_1 \boldsymbol{I}_p &\boldsymbol{0}&\cdots&\cdots&\boldsymbol{0}\\
      \boldsymbol{0} & T\lambda_2 \boldsymbol{I}_p&\boldsymbol{0}&\cdots&\boldsymbol{0}\\
      \boldsymbol{0}&\boldsymbol{0}&\ddots& \ddots& \boldsymbol{0}\\
      \boldsymbol{0}&\ddots&\ddots& \ddots& \boldsymbol{0}\\
      \boldsymbol{0}&  \cdots&\cdots&\boldsymbol{0}&T\lambda_n \boldsymbol{I}_p
    \end{matrix}\right]\,,
\end{equation}
where $\boldsymbol{I}_p$ indicates the $(p\times p)$ identity matrix, 
which provides the covariance matrix of the (vectorized) coefficients of matrix $\sqrt{K}\widehat{\boldsymbol{C}}_K$.

Therefore, for the normalized coupling matrix,
$\widehat{\boldsymbol{C}}_K\text{diag}(\sqrt{T\boldsymbol{\lambda}_0})^{-1}$, 
the column by column normalization, normalizes each block of the above covariance matrix by a multiplicative term $\frac{1}{T\lambda_k}$, 
to lead to an identity covariance. This proves
convergence of the normalized coupling matrix in distribution for $K\rightarrow +\infty$ to a random matrix with i.i.d. unit variance complex Gaussian coefficients 
(because uncorrelation implies independence in the Gaussian case)
\begin{equation}
  \label{eq:covVVmartSeq_finalDist}
  \sqrt{K}\mbox{Vec}\{\widehat{\boldsymbol{C}}_K \text{diag}(\sqrt{T\boldsymbol{\lambda}_0})^{-1}\}
  \underset{K\rightarrow +\infty
  }{\longrightarrow}\mathcal{N}\left(\boldsymbol{0}_{pn}, \rm I_{pn}\right)\,.
\end{equation}

\end{proof}


%% file: moments.tex
\begin{enumerate}
\item $\mathbb{E} \bar{X}_{jk}{X}_{lk}  = \delta_{lj}$, for all $k$,
\item $\frac{1}{n} \max_{j\neq l}\mathbb{E} \left| \bar{X}_{jk}{X}_{lk}\right|^2\rightarrow 0$ uniformly in $k\leq n$,
\item $\frac{1}{n^2}\sum_{\Gamma}\left(\mathbb{E}[\left( \bar{X}_{jk}{X}_{lk}-\delta_{lj}\right)\left( {X}_{j'k}\bar{X}_{l'k}-\delta_{j'l'}\right)\right)^2\rightarrow 0$ uniformly in $k\leq n$, where 
  $\Gamma \!=\!\{(j,l,j',l'):1\leq j,l,j',l'\leq p\}\!\setminus\! \{(j,l,j',l'): j\!=\!j'\!\neq\! l\!=\!l' \mbox{ or }j\!=\!l'\!\neq\! j'\!=\!l\}$,
\item $p/n\rightarrow \alpha \in (0,\infty)$.
\end{enumerate}

Based on the same developments as Theorem~\ref{thm:multiclt}, we use the auxiliary processes
\[
X_{lk}(t)=\frac{\sqrt{K}}{\sqrt{\lambda_kT}}\frac{1}{K}\int_{0}^{t}x_l(s)dP_k(s)
=\frac{1}{\sqrt{K\lambda_kT}}\int_{0}^{t}x_l(s)dP_k(s)=\int_{0}^{t} H_{lk}(s)dP_k(s)
\]
with $P_k$ zero-mean martingale associated to Poisson process of intensity $K\lambda_k$ 
(see Eq.(\ref{eq:multvarPoissonMartingale})) and 
\[
  H_{lk}(t)=\frac{x_l(t)}{\sqrt{K\lambda_kT}}\,,
\]
and will denote $X_{lk}=X_{lk}(T)$, i.e. the random variables that we are concerned with are the final values (at $t=T$) of those processes.

\textbf{Condition 1} is a direct application of results from Eq.(\ref{eq:covMult}) in the proof of Theorem~\ref{thm:multiclt}, because
$\mathbb{E} \left[\bar{X}_{jk}{X}_{lk}\right]$ is the covariance between the coefficients of the normalized coupling matrix.

\textbf{For condition 2}, let us first evaluate
\[
\mathbb{E}\left| \bar{X}_{jk}{X}_{lk}-\delta_{lj}\right|^2\,.
\]
For that we can use Ito's formula of Eq.(\ref{eq:itoJumpMult}) and derive the expression of $\bar{X}_{jk}{X}_{lk}$ as a stochastic integral, using the function 
$F(\bar{X}_{jk},{X}_{lk})= \bar{X}_{jk}{X}_{lk}$
We obtain
\begin{multline}\label{eq:productStochIntMart}
\bar{X}_{jk}{X}_{lk} =  -\int_0^T  \left({X}_{lk}\bar{H}_{jk}(s)+ \bar{X}_{jk}H_{lk}(s)\right)K\lambda_k ds \\
+ \int_0^T \!\!\! \left[\left(\bar{X}_{jk}(s_-)\!\!+\!\!\bar{H}_{jk}(s_-)\right)\left({X}_{lk}(s_-)\!\!+\!\!H_{lk}(s_-)\right)-\bar{X}_{jk}{X}_{lk}(s_-)\right]\left(dP_k(s)\!\!+\!\!K\lambda_k dt\right)\,,\\
\!\!\!= \!\!\!\int_0^T \!\!\! \left({X}_{lk}\bar{H}_{jk}(s_-)\!\!+\!\! \bar{X}_{jk}H_{lk}(s_-)\right)\!dP_k(s) 
\!+\!\!\! \int_0^T \!\!\! \left[\bar{H}_{jk}(s_-)H_{lk}(s_-)\right]\left(dP_k(s)\!\!+\!\!K\lambda_k ds\right)\!.
\end{multline}
The first term is a stochastic integral of a zero mean martingale, while the second term is a stochastic integral of a Poisson counting process, for which we can verify (due to Assumption~\ref{assum:multi}), that it has mean $\delta_{ij}$. As a consequence, $\mathbb{E}\left| \bar{X}_{jk}{X}_{lk}-\delta_{lj}\right|^2$ is the variance of the above expression, which is (by stochastic integral formula)
\begin{multline}
\mathbb{E}\left| \bar{X}_{jk}{X}_{lk}-\delta_{lj}\right|^2=
-\int_0^T  \mathbb{E}\left[\left({X}_{lk}(s_-)\bar{H}_{jk}(s_-)+ \bar{X}_{jk}(s_-)H_{lk}(s_-)\right)^2 \right]K\lambda_k ds \\
+ \int_0^T \left[\bar{H}_{jk}(s_-)H_{lk}(s_-)\right]^2K\lambda_k ds\,.
\end{multline}
Applying again the formula for predictable variation process, we obtain
\begin{multline}
\mathbb{E}\left| \bar{X}_{jk}{X}_{lk}-\delta_{lj}\right|^2=
-\int_0^T  \left[\int_{0}^{s}\left({H}_{lk}(u)\bar{H}_{jk}(s_-)+ \bar{H}_{jk}(u)H_{lk}(s_-)\right)^2 K\lambda_k du\right]K\lambda_k ds \\
+ \int_0^T \left[\bar{H}_{jk}(s_-)H_{lk}(s_-)\right]^2K\lambda_k ds\,.
\end{multline}
Due to Assumption~\ref{assum:multi}, this expression is bounded uniformly for any values of $i,j,n,k$, and condition 2 is fulfilled.

\textbf{For condition 3}, we use the auxiliary result presented in Proposition~\ref{prop:fourthMom} to compute the required fourth order moments.
\begin{multline*}
\frac{1}{K^2\lambda_k^2 }\mathbb{E}\left[\left( \bar{X}_{jk}{X}_{lk}\right)\left( {X}_{j'k}\bar{X}_{l'k}\right)\right]=
\int_{0}^{T} H_{lk} H_{j'k}ds \int_{0}^{T} \bar{H}_{jk} \bar{H}_{l'k}ds\\+\int_{0}^{T} H_{lk} \bar{H}_{jk} ds \int_{0}^{T} H_{j'k} \bar{H}_{l'k}ds +\int_{0}^{T} H_{lk}\bar{H}_{l'k}ds \int_{0}^{T} \bar{H}_{jk} H_{j'k} ds+\frac{1}{K\lambda_k}\int_{0}^{T} H_{lk} \bar{H}_{jk} H_{j'k} \bar{H}_{l'k}ds\\
=\frac{1}{\lambda_k^2 T^2 K^2}\left[
\int_{0}^{T} x_{l} x_{j'}ds \int_{0}^{T} \bar{x}_{j} \bar{x}_{l'}ds+\int_{0}^{T} x_{l} \bar{x}_{j} ds \int_{0}^{T} x_{j'} \bar{x}_{l'}ds +\int_{0}^{T} x_{l}\bar{x}_{l'}ds \int_{0}^{T} \bar{x}_{jk} x_{j'k} ds\right]\\
+\frac{1}{K^3\lambda_k^3T^2}\int_{0}^{T} x_{l} \bar{x}_{j} x_{j'} \bar{x}_{l'}ds
\,.
\end{multline*}
We first consider the term consisting in all products of two integrals, that we call \textit{integral product term}, the last term in this expression will be dealt with independently. Given Assumption~\ref{assum:multi}, it is clear that for $l,j,j',l'$ all different from each other, the integral product term is vanishing. If there happens to be only two indices that are equal, the moment also vanishes (at least one term of each product vanishes).
For the case $j=l=k'=l'$ the integral product term possibly does not vanish, but is uniformly bounded, and only $n$ terms satisfy this relation, such that it will not affect the limit of the relevant expression for condition 3 (due to the $1/n^2$ factor). 

Remains the case where three indices exactly are identical. In such case, one among  $\delta_{jl}$ or $\delta_{j'l'}$ is one while the other is zero. Take $\delta_{jl}=1$ and $\delta_{j'l'}=0$ without loss of generality, assuming $j=l=j'\neq l'$. The relevant quantity of condition 3 is
\begin{multline*}
\frac{1}{K^2\lambda_k^2 }\mathbb{E}\left[\left( \bar{X}_{jk}{X}_{lk}-1\right)\left( {X}_{j'k}\bar{X}_{l'k}\right)\right]=\frac{1}{K^2\lambda_k^2 }\mathbb{E}\left[\left( \bar{X}_{jk}{X}_{lk}\right)\left( {X}_{j'k}\bar{X}_{l'k}\right)\right]-\frac{1}{K^2\lambda_k^2 }\mathbb{E}\left[ {X}_{j'k}\bar{X}_{l'k}\right]\\
=\frac{1}{\lambda_k^2 T^2 K^2}\left[
\int_{0}^{T} x_{l} x_{j'}ds \int_{0}^{T} \bar{x}_{j} \bar{x}_{l'}ds+\int_{0}^{T}\left( x_{l} \bar{x}_{j}-T\right) ds \int_{0}^{T} x_{j'} \bar{x}_{l'}ds +\int_{0}^{T} x_{l}\bar{x}_{l'}ds \int_{0}^{T} \bar{x}_{jk} x_{j'k} ds\right]\\
+\frac{1}{K^3\lambda_k^3T^2}\int_{0}^{T} x_{l} \bar{x}_{j} x_{j'} \bar{x}_{l'}ds
\,,
\end{multline*}
in which, due to Assumption~\ref{assum:multi}, the integral product term still vanishes. As a consequence, the asymptotic behavior we are interested in is given by the behavior of the remaining \textit{single integral term} of the moment: $\frac{1}{K\lambda_k}\int_{0}^{T} x_{l} \bar{x}_{j} x_{j'} \bar{x}_{l'}ds$ (the only remaining non-vanishing terms are bounded and intervene only in $n$ terms of the sum), such that
\begin{multline}
\lim \frac{1}{n^2}\sum_{\Gamma}\left(\mathbb{E}\left[\left( \bar{X}_{jk}{X}_{lk}-\delta_{lj}\right)\left( {X}_{j'k}\bar{X}_{l'k}-\delta_{j'l'}\right)\right]\right)^2 =\\
\lim \frac{1}{n^2K^2\lambda_k^2}\sum_{\Gamma}\left(\mathbb{E}\left[\left( \bar{X}_{jk}{X}_{lk}-\delta_{lj}\right)\left( {X}_{j'k}\bar{X}_{l'k}-\delta_{j'l'}\right)\right]\right)^2 \,.
\end{multline}
Thus condition 3 is satisfied, due to the theorem's assumption.

To sum up, all four necessary condition for the application of Proposition~\ref{prop:cvMPdep} are fulfilled (condition 4 is part of the assumptions), and the convergence to the MP law follows immediately.


%% file: neuralComp_app_wishartEnsemble.tex

\label{sec:wishart}
Let $\boldsymbol{X}$ be a $p\times n$ data matrix.
Assume that coeffcient of $\boldsymbol{X}$, $x_{ij}$ are i.i.d. $\mathcal{N}_{\mathbb{C}}\left(0, 1\right)$.
$\mathcal{N}_{\mathbb{C}}$ specifies a \textit{standard complex normal} distribution.
By definition, this means that $x_{ij} = x_{ij}^{real} + \boldsymbol{i} x_{ij}^{imag}$,
where $x_{ij} = x_{ij}^{real}$ and $x_{ij}^{imag}$ are independent
(real) $\mathcal{N}(0,\frac{1}{2})$.
This implies that, columns of $\boldsymbol{X}$
are i.i.d. $\mathcal{N}_{\mathbb{C}}\left(\boldsymbol{0}_p, \rm I_p\right)$
and similary the real and imaginary parts are
$\mathcal{N}\left(\boldsymbol{0}_p, \rm I_p/2 \right)$.

As mentioned in main text,
as $n$ grows and $\frac{p}{n}\underset{n\rightarrow+\infty}{\rightarrow} \alpha \in (0,+\infty)$, 
the ESD of the so called Wishart ensemble $\boldsymbol{S}_n=\frac{1}{n}\boldsymbol{X} \boldsymbol{X}^H$, 
converges to the Marchenko-Pastur law $\mu_{MP}(x)$ \citep{marchenko1967distribution} with density
\begin{equation}
   \label{eq:mpLaw2}
  \frac{d\mu_{MP}}{dx}(x) = \frac{1-\alpha}{\alpha}\boldsymbol{1}_{\alpha>1}\delta_0+
    \frac{1}{2\pi \alpha x}\sqrt{(b-x)(x-a)}\boldsymbol{1}_{[a,b]}\,,
\end{equation}
with $a=(1 - \sqrt{\alpha})^2$ and $b=(1 + \sqrt{\alpha})^2$
(see examples for Marchenko-Pastur law for different values of $\alpha$ in Figure~\ref{fig:mpLawIlust}).

We wrote here the general formula that holds for all $\alpha>0$, accounting for zero eigenvalues with a Dirac mass in zero in the rank deficient case $\alpha>1$. 


%% file: fig2_MPlawIllust.tex
\begin{figure}
  \centering
  \includegraphics[width=.6\linewidth]{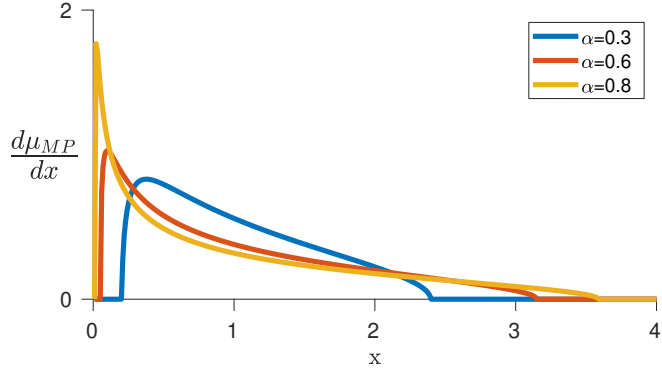}
  \caption{
    Density of the Marchenko-Pastur law for different values of the aspect ratio of the matrices, $\alpha$, in Eq.(\ref{eq:mpLaw}).     \label{fig:mpLawIlust} 
  }
\end{figure}


%% file: neuralComp_modifiedCoro5proof.tex
\begin{proof}[Proof]
  Similar to Corollary~\ref{corol:kappaplv}, we can derive the asymptotic $\rm PLV$ (Eq.(\ref{eq:theoPLV_simpleRateCoro5})) for this case,
  from the definition in Eq.(\ref{thePLV}).
  We and use the assumed phase $\phi(t)=2\pi kt/T$ 
  and apply the intensity function defined in Eq.(\ref{eq:simpleRateCoro5})
  ,
  in Corollary~\ref{corol:1DPLV},
  \begin{align}
    {\rm PLV^*}
    & =  \frac{\int_0^T e^{\boldsymbol{i}\phi(t)}\lambda(t)dt}{\int_0^T \lambda(t)dt}\\
      & =\frac
      {\int_{0}^{T} e^{\boldsymbol{i}2\pi kt/T} 
      \left(1+ \varkappa \cos\left( 2\pi m t / T-\varphi_0 \right)  \right)
      dt}
      {\int_{0}^{T} 
      \left(1+ \varkappa \cos\left( 2\pi m t / T-\varphi_0 \right)  \right)
      dt}\,.
  \end{align}
  By using the Euler's formula we can write the second term in the numerator as weighted sum of exponentials ($cos(x) = \frac{1}{2}(e^{\boldsymbol{i} x} + e^{-\boldsymbol{i} x}))$, 
  \begin{align}
    {\rm PLV^*}
    & = \frac{1}{2} \frac
      {\int_{0}^{T} e^{\boldsymbol{i}2\pi kt/T} 
      + \varkappa \int_{0}^{T} e^{\boldsymbol{i}2\pi kt/T}
      \left( e^{\boldsymbol{i} (2\pi m t / T-\varphi_0)} + e^{-\boldsymbol{i} (2\pi m t / T-\varphi_0)  }  \right) 
      dt}
      {\int_{0}^{T} dt +
      \int_{0}^{T} \varkappa \cos\left( 2\pi m t / T-\varphi_0 \right) dt}\\
    & = \frac{1}{2} \frac
      {\int_{0}^{T} e^{\boldsymbol{i}2\pi kt/T} 
      + \varkappa \int_{0}^{T} e^{\boldsymbol{i}2\pi (k+m)t/T}e^{\boldsymbol{i}\varphi_0}
      + \varkappa \int_{0}^{T} e^{-\boldsymbol{i}2\pi (k-m)t/T}e^{\boldsymbol{i}\varphi_0}
      dt}
      {\int_{0}^{T} dt +
      \int_{0}^{T} \varkappa \cos\left( 2\pi m t / T-\varphi_0 \right) dt}\\
    & = \frac{1}{2} \frac
      {\int_{0}^{T} e^{\boldsymbol{i}2\pi kt/T} 
      + \varkappa e^{\boldsymbol{i}\varphi_0}  \int_{0}^{T} e^{\boldsymbol{i}2\pi (k+m)t/T}
      + \varkappa e^{\boldsymbol{i}\varphi_0} \int_{0}^{T} e^{-\boldsymbol{i}2\pi (k-m)t/T}
      dt}
      {\int_{0}^{T} dt +
      \varkappa \int_{0}^{T}  \cos\left( 2\pi m t / T-\varphi_0 \right) dt}\,.
  \end{align}
  Given that $k,m > 0$ and we are integrating over full periods all terms vanishes
  except the last term in the numerator (if and only if $k=m$) and first term in the denominator. Therefore we have,
  \begin{align}
    {\rm PLV^*}
    & = \frac{1}{2} \frac
      {\varkappa e^{\boldsymbol{i}\varphi_0} \int_{0}^{T} e^{-\boldsymbol{i}2\pi (k-m)t/T} dt}
      {\int_{0}^{T} dt}\\
    & = \frac{1}{2}\varkappa e^{\boldsymbol{i}\varphi_0} \delta_{km}\,.
  \end{align}
  We prove the first part of the corollary.

  The covariance matrix of the asymptotic distribution,
  can be derived by the procedure we used for the proof of
  Corollary~\ref{corol:kappaplv}.
  We plug the rate $\lambda(t)$ assumed in the corollary (Eq.(\ref{eq:simpleRateCoro5}))
  and integrate on $[0,\,T]$,
  \begin{equation}
    \left(\rm Cov(Z) \right)_{11} 
    = \frac{\lambda_0}{\Lambda(T)^2} \int_0^T \cos^2(2\pi kt/T)
      \left(1 + \varkappa \cos\left( 2\pi m t / T-\varphi_0 \right)  \right) dt\,,
  \end{equation}
  and use the trigonometric identity Eq.(\ref{eq:trigIdCos}), to get
  \begin{equation}
    \left(\rm Cov(Z) \right)_{11} 
    = \frac{\lambda_0}{2\Lambda(T)^2}
    \int_0^T \left( 1 + \cos(4\pi kt/T) \right)
    \left(1 + \varkappa \cos\left( 2\pi m t / T-\varphi_0 \right)  \right) dt\,.
  \end{equation}
In the resulting equation,
  \begin{equation}
    \begin{split}
      \left(\rm Cov(Z) \right)_{11}
      =  \frac{\lambda_0}{2\Lambda(T)^2}
      \left[
        \int_0^T dt  
        + \varkappa \int_0^T \cos\left( 2\pi m t / T-\varphi_0\right) dt  \right. &
      + \int_0^T \cos(4\pi kt/T) dt \\
      \left. + \varkappa \int_0^T \cos(4\pi kt/T) \cos\left( 2\pi m t / T-\varphi_0\right) dt  
      \right]
    \end{split}
  \end{equation}
  all terms vanish, except the first term.
  Second and third vanishes as we integrate in the full period and 
  the last term vanishes given that,
    \begin{multline}
      \int_0^T \cos(4\pi kt/T) \cos\left( 2\pi m t / T-\varphi_0\right) dt \\
      = \cos(\varphi_0) \int_0^T
      \cos(4\pi kt/T)
        \cos\left( 2\pi m t / T) \right) dt
        + \sin(\varphi_0) \int_0^T \cos(4\pi kt/T) \sin\left( 2\pi m t / T)\right)
      dt \,,
  \end{multline}
  and $k$ and $m$ are integers.

  Finally, given that $\Lambda(T) = \int_0^T \lambda(t) dt  =  \lambda_0 T$, we have, 
  \begin{equation}
          \left(\rm Cov(Z) \right)_{11} = \frac{1}{2 \lambda_0 T}\,.
  \end{equation}

  We have a similar calculation for the imaginary part \ie $\left(\rm Cov(Z) \right)_{22}$.
  The off-diagonal elements of the covariance matrix vanish due to symmetry of integrand.

  Therefore, we showed that for the scaled residual,
  \[
    Z =  \sqrt{K}\left(\widehat{\rm PLV}_K - \rm PLV^* \right)\,,
  \]
  converges to a zero mean isotropic complex Gaussian:
  \begin{equation*}
    \sqrt{K}\left(\widehat{\rm PLV}_K - \rm PLV^* \right)
    \underset{K\rightarrow +\infty
    }{\longrightarrow}\mathcal{N}\left(
      \left[\begin{matrix}
          0\\
          0
        \end{matrix}\right], \frac{1}{2 \lambda_0 T}\left[
        \begin{matrix}
          1&0\\
          0&	1
        \end{matrix}\right]\right)\,.
  \end{equation*}
\end{proof}
